\documentclass[]{IEEEtran}
\usepackage[utf8]{inputenc}
\usepackage[english]{babel}
\usepackage{amsthm}
\usepackage{amsfonts}%
\usepackage{amsmath}%
\usepackage{fixmath}%
\ifCLASSINFOpdf
\usepackage{caption}
\usepackage{subfloat}
\usepackage[pdftex]{graphicx}
\graphicspath{ {./plots/} }
\else
\fi
\usepackage{amsmath,amsfonts,amssymb,array}

\theoremstyle{definition}

\theoremstyle{remark}

\usepackage{algorithmic}
\usepackage{array}
\usepackage[table,xcdraw]{xcolor}
\usepackage{cite}
\usepackage{siunitx}
\allowdisplaybreaks
\usepackage{derivative}

\usepackage{amsthm}
\usepackage{amsfonts}%
\usepackage{amsmath}%
\usepackage{fixmath}%
\usepackage{amssymb}%
\usepackage{graphicx}
\usepackage{comment}
\usepackage{url}
\usepackage{cite}
\usepackage{xcolor}
\usepackage{amsthm}
\usepackage{subfloat}
\usepackage{subfig}
\usepackage{flushend}

\newtheorem{theorem}{Theorem}

\usepackage{mathtools}

\begin{document}
	\title{MIMO Radars and Massive MIMO\\Communication Systems can Coexist}
	\author{\IEEEauthorblockN{Aparna Mishra, Ribhu Chopra} 
		
		\thanks{A. Mishra and R. Chopra are with the Department of Electronics and Electrical Engineering, Indian Institute of Technology Guwahati, Assam, India.  (emails: m.aparna@iitg.ac.in, ribhu@outlook.com).
			
		}
	}
	\maketitle		
	\begin{abstract}
	  In this paper, we investigate the coexistence of a single cell massive MIMO communication system with a MIMO radar. We consider the case where the massive MIMO BS is aware of the radar's existence and treats it as a non-serviced user, but the radar is unaware of the communication system's existence and treats the signals transmitted by both the BS and the communication users as noise. Using results from random matrix theory, we derive the rates achievable by the communication system and the radar. We then use these expressions to obtain the achievable rate regions for the proposed joint radar and communications system. We observe that due to the availability of a large number of degrees of freedom at the mMIMO BS, results in minimal interference even without co-design. Finally we corroborate our findings via detailed numerical simulations and verify the validity of the results derived previously under different settings. 
	\end{abstract}
	
	\section{Introduction}

    Recently, the design of joint radar and communications~(JRC) systems, especially that of jointly designed communication and sensing~(JCAS) systems has become an active area of research~\cite{NHTSA,need_of_coexistence_survey_2017,Liu_JSAC_2022}. This is due to the improved spectral efficiencies and hardware costs offered by these systems. Moreover, these systems are seen as key enablers for the paradigm of intelligent transportation systems~(ITS), popularly known as smart vehicles. In general, there exist three approaches towards the design of JRC systems, viz.  coexistence, cooperation, and co-design~\cite{chiriyath2017radar}. As apparent from the name, the coexistence based approach considers the performance of the two subsystems designed separately treating the signals from each other as interference. In this case the two subsystems may or may not be aware of each other's existence. In the cooperation based approach, the two systems might be designed separately but are aware of each others presence and cooperate to mitigate interference. Finally, a co-designed JRC system considers a scenario where the two subsystems are designed jointly to maximize each other's performance. 
It is apparent that while the co-design based approach optimizes the performances of both the underlying subsystems it is unsuited for most legacy hardware and necessitates a hard reboot of the system architecture. On the other hand, the coexistence based approach is the least disruptive approach for legacy hardware. Therefore, in this paper we focus on the coexistence based design of a JRC system~\cite{liu2020joint}.

Over the last decade, the idea of using a large number of antennas in both radar and communications systems, dubbed massive multiple input multiple output~(MIMO), has also gained much traction~\cite{intro_MMIMO,myths,marzetta2016fundamentals,mMIMO_radar}.  While the literature dealing with radar technologies has focused more on conventional MIMO radars~\cite{mahal2017spectral}, with limited focus on massive MIMO~\cite{mMIMO_radar}; massive MIMO has been established as a front runner technology for next generation wireless communication systems. 
 Massive MIMO systems have been shown to be resilient to jamming~\cite{uplink_mMIMO_2020} and other forms of in-band interference making them ideal candidates for sharing spectrum with radars. These features, coupled with the fact a BS with a large antenna array can easily form a null to minimize interference to a co-existent radar subsystem make massive MIMO communication systems ideal candidates for coexistence based JRC systems. Consequently, in this paper we study the performance of a JRC system where a MIMO radar co-exists with massive MIMO communication system.

  	\subsection{Related Work}
 In the context of coexistence based JRC systems, the authors in~\cite{Reed_2016_WCL} and~\cite{bell_2014_ICEAA} have experimentally demonstrated the detrimental effects of the presence of a proximal in-band radar on communications systems. In~\cite{Rathapon_2012_selectedareasincomm} the idea of opportunistic spectrum sharing  between a rotating radar, and a cognitive communication system is analysed. 
The authors in~\cite{Sodagari_2012_NSP_GLOBECOM,Khawar_IEEE_Sensors_Journal_2015} have considered the use of null space projection~(NSP) for achieving radar communication coexistence. Under this approach the system with a larger number of degrees of freedom projects its signal onto the null space of the interference channel to the system with smaller number of degrees of freedom. The present literature pertaining to NSP based JRC mostly focuses on MIMO radars having a larger number of degrees of freedom preventing interference to SISO systems~\cite{Sodagari_2012_NSP_GLOBECOM,Babaei_2013_GLOBECOM_NSPexpansion}.

  Since coexistence based JRC systems witness a performance degradation in both radar and communication subsystems, it becomes important to quantify these losses on the same scale in order to appropriately identify the underlying trade-offs. Consequently, authors in~\cite{chiriyath2015inner} have introduced the idea of ``radar information rate" as an analogue of the achievable rate of a communication system. The formulation of the radar information rate models the target as an unwilling source of information, with the radar receiver acting as the sink. The radar information rate is then defined as the mutual information between the unwilling source and the sink. Consequently, a JRC system can be viewed as multiple access channel~(MAC)~\cite{cover1999elements} consisting of a radar subsystem and a communication subsystem whose overall performance can be analysed in terms of the achievable rate regions of these two subsystems~\cite{chiriyath2015inner,chiriyath2017radar,Yu_SAM_2018,Chiri_TAE_2019,Bliss_RadCon_2014}.  Alternatively, the authors in \cite{guerci2015joint} have considered each resolution cell of the radar as a constellation point and defined the "channel capacity" of the radar as the maximum information contained in the echo signal. However, in this work we evaluate the performance of the JRC system in terms of rate regions between the achievable communication rate and the radar information rate due to the simplicity and intuitiveness of this approach. The performance of massive MIMO communications systems is mostly quantified in terms of the per user achievable rates~\cite{Chopra_TWC_2018} that are well characterized by the logarithms of the deterministic equivalents~(DEs) of their signal to interference-plus-noise ratios~(SINRs)~\cite{RMT,Kolomvakis2017massive}. In this paper as well, we analyse the performance of the underlying massive MIMO system in terms of achievable rates obtained via DE analysis. 
  
  It is also important to note that despite the recent advances in millimetre wave communication technologies, the sub-6~GHz spectrum remains preferred for long distance communications, much of which has been dedicated for radar usage~\cite{Opportunistic_Sharing_2015}. Therefore, our system model is built around the sub-6~GHz rich scattering channel model~\cite{Jakes}.
  

  The idea of massive MIMO enabled JRC has recently been explored in the literature~\cite{Buzzi_Asilomar_2019,Buzzi_ITG_2020,uplink_mMIMO_2020}. Out of these only~\cite{uplink_mMIMO_2020} discusses a coexistence based JRC system. However, the underlying analysis  is limited to studying the effect of radar interference on the uplink of a massive MIMO system. In contrast, in this paper, we analyse the problem  from  perspectives of both the communication system and the radar, using the achievable rate regions as a performance metric over an entire communication frame (i.e. both uplink and downlink). Also, instead of the use-and-then-forget bounds~\cite{marzetta2016fundamentals} used in~\cite{uplink_mMIMO_2020} for evaluating the performance of massive MIMO systems, we use deterministic equivalent analysis~\cite{Papa_TVT_2016}, that results in better approximations of achievable SINRs for MMSE type receivers. 
	\subsection{Contributions}

  In this paper, we use rate regions to characterize the performance of a coexistence based JRC system comprising a single cell massive MIMO communication system and a static MIMO radar over a full communication frame. Our specific contributions are enumerated as follows:
    \begin{enumerate}
    	\item We first determine the channel estimation performance of the massive MIMO system with uplink training in the presence of radar generated interference, and obtain expressions for consequent channel estimation mean squared error~(MSE). Similarly, we obtain the MSE of the angle of arrival estimate at the radar in the presence of the interference generated by the communication system. These expressions are then used to characterize a trade-off between the pilot powers employed by the users and the radar transmit power~(See Section~III.).
    	\item Following this, using DE analysis, we obtain expressions for uplink achievable rates using MMSE combining at the BS in the presence of radar generated interference. We then derive the Cramer Rao Bound~(CRB) on the angle of arrival (AoA) estimate at the radar, in the presence of interference caused due to uplink transmission by the users, and use it to form an upper bound on the radar rate~(See Section~IV.).  
    	\item We then derive the DEs for the downlink SINR at the users, assuming regularized zero forcing~(RZF) beamforming at the BS, in the presence of radar generated interference. We assume that the RZF beamforming at the BS also forms a null in the direction of the available estimate of the radar channel, and use this information to calculate the CRB on the AoA estimation performance and the corresponding radar rates~(See Section~V.).
    	\item Via extensive numerical simulations we validate our derived results, and plot the achievable rate regions for our JRC system for various use cases. We find that the availability of a large number of degrees at the BS results in minimal interference to both the constituents of the JRC systems, resulting in a significantly convex rate regions~(See Section~VI.).
    \end{enumerate}
We can thus conclude that coexistence based design of JRC systems is possible in the massive MIMO regime, thus allowing for the addition of sensing capabilities to legacy systems, without the need for an extensive redesign. We next describe the system model considered in this work.

\textcolor{black}{\subsection{Notation}
Throughout this paper, lower case letters indicate scalar quantities, lower case $(\mathbf{a})$ and the upper case $(\mathbf{H})$ bold face characters respectively indicate column vectors and matrices. $ (.)^H, (.)^T, (.)^*,$ respectively represent Hermitian, transposition and complex conjugation operation on vectors and matrices. $\Re\{.\}$ and $\Im\{.\}$ respectively represent the real and imaginary parts of a complex number. $\text{Tr}(.)$ denotes the trace of a matrix. $ \underset{M \rightarrow \infty}{\overset{a.s.}\longrightarrow}$ indicates almost sure convergence. $\text{vec}(\mathbf{A})$ represents the vectorization of matrix $\mathbf{A}$, $\text{diag}( \mathbf{a})$ represents the diagonalization of vector $\mathbf{a}$, $E[.]$ represents the expectation operator and $\mathbf{I}_M$ represents the identity matrix of order $M$. $M, N_t $ and $N_r$ denote the number antennas at the BS, radar transmitter and radar receiver respectively. $K$ denotes the number of UEs. $N$ denotes the total transmission time instants. $\sigma_r^2$, $\epsilon_{u,s,k}$ and $\epsilon_{d,s,k}$ denote the transmit powers of the radar, the $k$th UE and $k$th BS antenna respectively. $\beta$ denotes the large scale fading coefficient. $\mathcal{CN}(\boldsymbol{\mu},\boldsymbol{\Sigma})$ represents a circularly symmetric complex Gaussian~(CSCG) random vector with mean vector $\boldsymbol{\mu}$  and covariance matrix $\boldsymbol{\Sigma}$.} 

\section{System Model}
We consider an RCC system comprising a single cell massive MIMO subsystem coexisting with an in-cell MIMO radar as shown in Fig.~\ref{f1}. The two subsystems are assumed to transmit over the same time frequency resources with a full bandwidth overlap, but are assumed to have only non line of sight~(NLoS), rich scattering interference channels. We next describe the system and signal models for the two subsystems individually. 
 \begin{figure} [t]
	\centering
	\includegraphics[width= \linewidth]{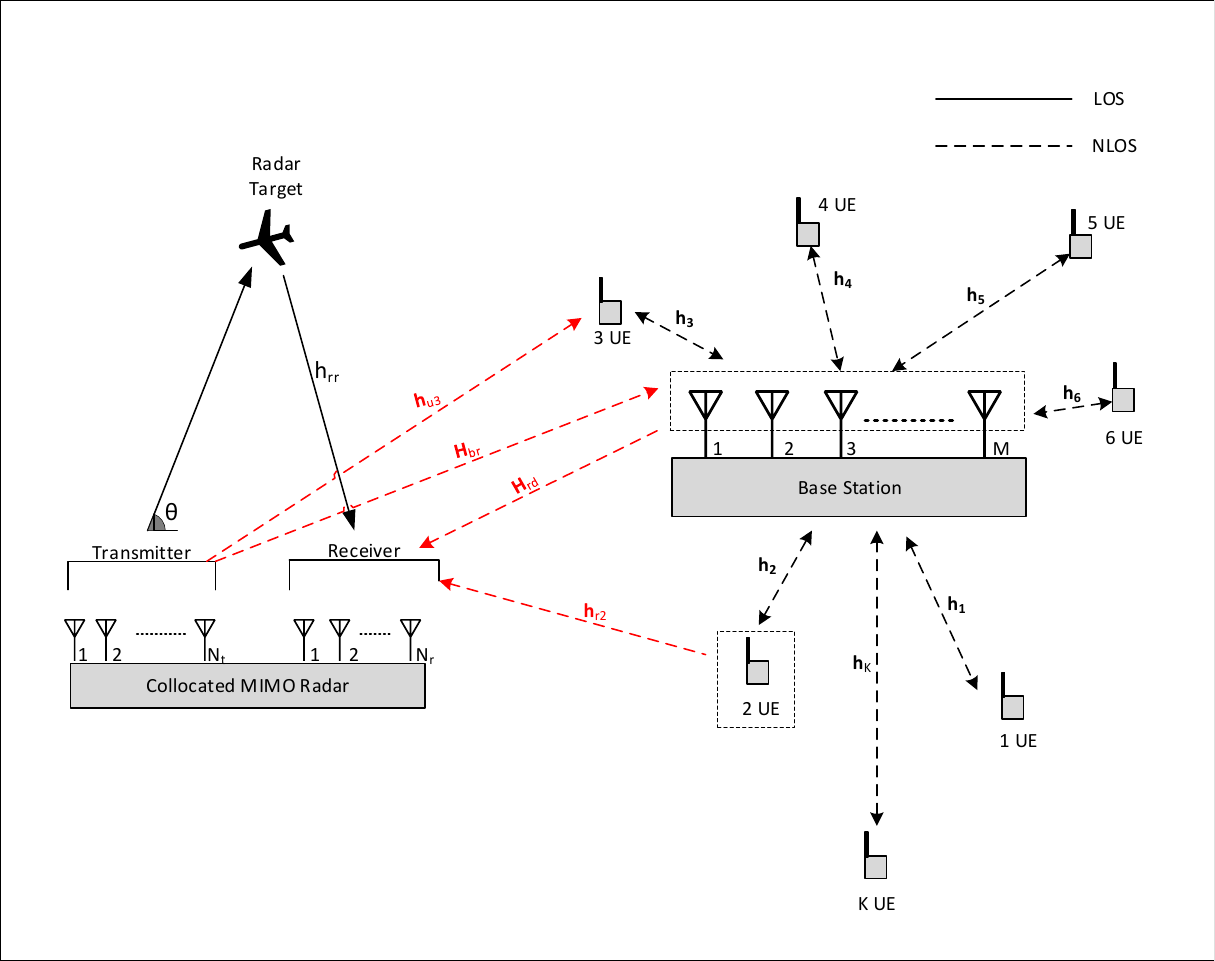}
	\caption{The System Model}
	\label{f1}
\end{figure}
 \subsection{The Radar Subsystem}
 We consider a mono-static pulsed MIMO radar equipped with collocated transmit and receiver antenna arrays, as shown in Fig.~\ref{f1}. The radar subsystem consists of $N_t$ transmit antennas and  $N_r$ receive antennas. Since the radar is mono-static, the transmit and receive antenna arrays can safely be assumed to be synchronized, allowing for coherent processing of transmit and receive signals. We assume that a single target is present in the LOS of the radar, at an angle $\theta$, such that the array response vectors of the transmit and receive arrays are respectively given by $\mathbf{a}_t^T(\theta)$ and $\mathbf{a}_r^T(\theta)$. We also let $h_{rr}$ denote the reflection coefficient of the said target that accumulates the effects of propagation attenuation, phase shifts, and the radar cross section of the target. Now, let $\mathbf{s}[n] \in \mathbb{C}^{N_{t} \times 1}$  be the signal transmitted by the radar at the $n$th instant, such that $E[\mathbf{s}[n]\mathbf{s}^H[n]]=\mathbf{R}_{ss}=\sigma_r^2\mathbf{I}_{N_t}$, with $\mathbf{I}_K$ representing the order $K$ identity matrix. Then, the signal received by the radar, denoted by $\mathbf{z}[n]\in\mathbb{C}^{N_r \times 1}$, in the absence of any communication interference and in multipath free propagation~\cite{li2007mimo}, can be expressed as,
 \begin{equation}
 			\mathbf{z}[n] =  h_{rr} \mathbf{A}(\theta) \mathbf{s}[n]+\sqrt{N}_0 \mathbf{w}_r[n], \\
  \end{equation}
 where $\mathbf{w}_r[n]$ denotes the temporally and spatially white, zero mean circularly symmetric complex Gaussian~(ZMCSCG) additive noise, and $\mathbf{A}(\theta)=\mathbf{a}_r(\theta) \mathbf{a}^T_t(\theta)$. We also assume that the target is moving slowly with respect to the radar, and we can ignore the Doppler shift within a pulse~\cite{sun2015mimo}.   
 
 \subsection{The Communication Subsystem}  
 We consider a single cell massive MIMO communication system operating in the time division duplexed~(TDD) mode with a BS equipped with $M$ antenna elements serving $K$ single antenna user equipments~(users). We assume the channels between the BS and users to be \textcolor{black}{correlated} with frequency flat rich scattering. We let $\sqrt{\beta_k} \mathbf{h}_{k} \in \mathbb{C}^{M \times 1} $ denote the channel vector between the BS and the $k${th} user with $\beta_k$ and  \textcolor{black}{$\mathbf{h}_{k} \sim \mathcal{CN} (\mathbf{0},\mathbf{\Sigma}_{k})$} representing the large scale and small scale fading coefficients, respectively, \textcolor{black}{where $\mathbf{\Sigma}_{k} \in \mathcal{C}^{M \times M}$ is the covariance matrix of $\mathbf{h}_{k}$}. 
 
The communication frame is divided into three sub-frames, viz. channel estimation, uplink data transmission, and downlink data transmission. In the first sub-frame, spanning $K$ channel uses, the users transmit orthogonal pilot signals that are received by the BS and are used to form MMSE estimates of the BS to user channels. Following this, during the second sub-frame, spanning $\tau_u$ channel uses, the users transmit uplink data, and the BS uses the available channel estimates to effectively decode this data via MMSE combining~\cite{Chopra_TWC_2018}. Finally, during the downlink data transmission sub-frame spanning $\tau_d$ channel uses, the BS, under the assumption of channel reciprocity~\cite{Chopra_TSP_2020,Papazafeiropoulos2017ageing}, uses the available channel estimates to appropriately beamform and transmit data to the users. We next describe the signal models for these sub-frames.  
 \subsubsection{Channel Estimation} Let the $k$th user transmit a pilot signal $\psi_k[n]$ for $n\in[1,K]$, with an energy $\epsilon_{u,p,k}$ such that $\sum_{n=1}^{K}\psi_k[n]\psi^*_l[n]=\delta[k-l]$, with $\delta[n-k]$ representing the Kronecker delta function. Then, the signal vector received by the BS at the $n$th instant without accounting for radar interference can be expressed as
    \begin{equation}
	\mathbf{y}[n]= \sum_{k=1}^{K} \sqrt{\beta_k\epsilon_{u,p,k}} \mathbf{h}_{k} \psi_k[n]+\sqrt{N_0} \mathbf{w}_b[n],
\end{equation}  
 with  $\mathbf{w}_b[n]$ representing the temporally and spatially white ZMCSCG additive noise with unit variance.  
 \subsubsection{Uplink Transmission} Letting the $k$th user transmit the data symbol $x_k[n]$ ($E[x_k[n]x_l[m]^*]=\delta[n-m]\delta[k-l]$) at the $n$th instant with an energy $\epsilon_{u,s,k}$, we can write the signal received at the BS in the absence of any radar interference as, 
     \begin{equation}
 	\mathbf{y}[n]= \sum_{k=1}^{K} \sqrt{\beta_k\epsilon_{u,s,k}} \mathbf{h}_{k} x_k[n]+\sqrt{N_0} \mathbf{w}_b[n].
 \end{equation}
\subsubsection{Downlink Transmission} Letting $\mathbf{Q} \in \mathcal{C}^{M \times K}$ denote the preceding matrix at the BS, $\epsilon_{d,s,k}$ the downlink symbol energy for the $k$th user such that the corresponding symbol sent to the $k$th user at the $n$th instant is $p_k[n]$, we can write the downlink signal received at the $k$th user at the $n$th instant as
\begin{equation}
	r_k[n]=\mathbf{h}_k^T\mathbf{Q}\text{diag}(\sqrt{\boldsymbol{\epsilon_{d,s}}})\mathbf{p}[n]+\sqrt{N_0}w_{k}[n],
\end{equation} 
where $\mathbf{p}[n]=[p_1[n],p_2[n],\ldots,p_K[n]]^T$, and $\boldsymbol{{\epsilon}_{d,s}}=[{\epsilon}_{d,s,1},\ldots,{\epsilon}_{d,s,K}]^T$, and $w_{k}[n]$ is the ZMCSCG noise with unit variance at the $k$th user.

  \subsection{Interference Channel Models}
  As stated earlier, the MIMO radar and massive MIMO communication sub-systems do not have line of sight interference channels. Since both the BS and Radar are fixed, the interference channels between them, denoted by $\mathbf{G}_{rb}\in\mathbb{C}^{M\times N_t}$ and  $\mathbf{G}_{br}\in\mathbb{C}^{N_r\times M}$, respectively, for the radar transmit and receive arrays are also assumed to be time invariant. Now, in accordance with the rich scattering assumption, the entries of $\mathbf{G}_{rb}$ and $\mathbf{G}_{rb}$ are assumed to be independent and identically distributed~(i.i.d.) ZMCSCG with a variance $\eta_{I}=\min(d^{-\alpha}_{br},1) $ with $d_{br}$ being the distance between the BS and the radar. Also, their MMSE estimates, respectively, given by  $\hat{\mathbf{G}}_{rb}$ and $\hat{\mathbf{G}}_{br}$ are assumed to be available at the BS such that,
  \begin{equation}
  	{\mathbf{G}}_{rb}=\hat{\mathbf{G}}_{rb}+\tilde{\mathbf{G}}_{rb},\quad
	{\mathbf{G}}_{br}=\hat{\mathbf{G}}_{br}+\tilde{\mathbf{G}}_{br},
\end{equation}
where $\tilde{\mathbf{G}}_{rb}$ and $\tilde{\mathbf{G}}_{br}$ are estimation errors orthogonal to $\hat{\mathbf{G}}_{rb}$ and $\hat{\mathbf{G}}_{br}$, respectively, and their entries have a variance $\eta_{e}$. \textcolor{black}{In case the BS is unable to estimate the interference channels, then $\hat{\mathbf{G}}_{rb}$ and $\hat{\mathbf{G}}_{br}$ are set to zero and $\eta_e=\eta_I$.} During the training and uplink data transmission phases, the BS uses $\hat{\mathbf{G}}_{rb}$ to form nulls in the directions containing the radar interfering signals to minimize their effects on channel estimation and uplink performances. Similarly, during the downlink transmission, the BS forms nulls in the direction of $\hat{\mathbf{G}}_{br}$ to minimize the interference to the radar.

Similarly, the interference channel between the radar transmit array and the $k$th user is represented by $\mathbf{g}_{rk}\in\mathbb{C}^{N_t \times 1}$, $k\in\{1,2,\ldots,K\}$, and the interference channel between the $k$th user and the radar receive array is represented by $\mathbf{g}_{kr}\in\mathbb{C}^{N_r \times 1}$, $k\in\{1,2,\ldots,K\}$. Both $\mathbf{g}_{rk}$ and $\mathbf{g}_{kr}$ are assumed to consist of i.i.d. ZMCSCG entries having a variance $\eta_{rk}$ that is equal to the large scale fading coefficient between the radar and the $k$th user. 
     
    \section{The Channel Estimation Sub-frame}
    \subsection{The Communication Subsystem}
     Considering the effect of radar generated interference, the signal received by BS antennas at the $n$th instant, denoted by, \textcolor{black}{$\mathbf{y}[n] \in \mathcal{C}^{M \times 1}$} can be expressed as,
    \begin{equation}
   \textcolor{black}{ \mathbf{y}[n]= \sum_{k=1}^{K} \sqrt{\beta_k\epsilon_{u,p,k}} \mathbf{h}_{k}  \psi_k[n]+  \mathbf{G}^H_{rb}\mathbf{s}[n]+\sqrt{N_0} \mathbf{w}[n],}	
    \end{equation}
\textcolor{black}{Defining $\mathbf{y}_{l} \triangleq \sum_{n=1}^{K} \mathbf{y}[n] \psi^*_l[n]$, we obtain
    \begin{multline}
     \mathbf{y}_{l}  = 	\sqrt{\beta_l \epsilon_{u,p,l}}	\mathbf
     h_{l} + \sum_{n=1}^{K} \mathbf{\hat{G}}^H_{rb,}  \mathbf{s}[n] \psi^*_l[n] \\+\sum_{n=1}^{K} \mathbf{\tilde{G}}^H_{rb}  \mathbf{s}[n] \psi^*_l[n]
     + \sum_{n=1}^{K} \sqrt{N_0} \mathbf{w}[n] \psi^*_l[n].
\end{multline}
 Now, the BS may or may not have the phase synchronization information for $\mathbf{s}[n]$. In the former case $\mathbf{y}'_{l}$ can be formed by subtracting $\sum_{n=1}^{K} \mathbf{\hat{G}}^H_{rb}  \mathbf{s}[n] \psi^*_l[n]$ from $\mathbf{y}_{l}$ as 
 \begin{equation}
 	\mathbf{y}'_{l}  =\sqrt{\beta_l \epsilon_{u,p,l}}	\mathbf{h}_{l}
 	+ \sum_{n=1}^{K}  \tilde{\mathbf{G}}^H_{rb}\mathbf{s}[n] \psi^*_l[n]+ \sum_{n=1}^{K} \sqrt{N_0} \mathbf{w}[n] \psi^*_l[n].	
 \end{equation}
Clearly, in this case, the LMMSE estimate $\hat{\mathbf{h}}_{l}$ of $\mathbf{h}_{l}$ can be written as $\mathbf{A}'_{l}\mathbf{y}'_{l}$, with $
\mathbf{A}'_{l}  = E[\mathbf{h}_{l} \mathbf{y}'^{H}_{l}] [E[\mathbf{y}'_{l} \mathbf{y}'^{H}_{l}]]^{-1}, 
$
such that $E[\mathbf{h}_{l} \mathbf{y}^{H}_{l}]=E[\mathbf{h}_{l} \mathbf{y}'^{H}_{l}]=\sqrt{\beta_l \epsilon_{u,p,l}} \mathbf{\Sigma}_{l}$, and $E[\mathbf{y}'_{l} \mathbf{y}'^{H}_{l}]=\beta_l \epsilon_{u,p,l} \mathbf{\Sigma}_{l} +(N_t\eta_e\sigma_r^2+N_0) \mathbf{I}_{M}$. Similarly, in case the synchronization information about $\mathbf{s}[n]$ is not available at the BS, $\hat{\mathbf{h}}_{l}=\mathbf{A}_{l}\mathbf{y}_{l}$, such that
\begin{multline}
	\mathbf{A}_{l} =  E[\mathbf{h}_{l} \mathbf{y}^{H}_{l}] [E[\mathbf{y}_{l} \mathbf{y}^{H}_{l}]]^{-1}  \\= \beta_l \epsilon_{u,p,l} \mathbf{\Sigma}_{l} [\beta_l \epsilon_{u,p,l} \mathbf{\Sigma}_{l}+\sigma_r^2 \mathbf{\hat{G}}_{rb} \mathbf{\hat{G}}^H_{rb} +(N_t\eta_e\sigma_r^2+N_0) \mathbf{I}_{M}]^{-1}.
\end{multline}
Now, letting $\tilde{\mathbf{h}}_{l}$ represent the ZMCSCG estimation error orthogonal to $\hat{\mathbf{h}}_{l}$, it is easy to show that $\mathbf{h}_{l}$ can be represented as
\begin{equation}
	\mathbf{h}_{l}=\hat{\mathbf{h}}_{l}+\tilde{\mathbf{h}}_{l},
\end{equation}
with $\hat{\mathbf{h}}_{l}$ and $\tilde{\mathbf{h}}_{l}$ having covariance matrices  $\mathbf{B}_{l}$ and $\bar{\mathbf{B}}_{l}$ respectively. Here, $$\mathbf{B}_{l}=\beta_l \epsilon_{u,p,l} \mathbf{\Sigma}_{l} [\beta_l \epsilon_{u,p,l} \mathbf{\Sigma}_{l} +(N_t\eta_e\sigma_r^2+N_0) \mathbf{I}_{M}]^{-1}  \mathbf{\Sigma}_{l} $$ and $$\bar{\mathbf{B}}_{l}=\mathbf{\Sigma}_{l} - \beta_l \epsilon_{u,p,l} \mathbf{\Sigma}_{l} [\beta_l \epsilon_{u,p,l} \mathbf{\Sigma}_{l} +(N_t\eta_e\sigma_r^2+N_0) \mathbf{I}_{M}]^{-1}  \mathbf{\Sigma}_{l},$$ when the radar signal synchronization information is known at the BS, and
 $$ \mathbf{B}'_{l}=\beta_l \epsilon_{u,p,l} \mathbf{\Sigma}_{l} [\beta_l \epsilon_{u,p,l} \mathbf{\Sigma}_{l} +\sigma_r^2 \mathbf{\hat{G}}_{rb} \mathbf{\hat{G}}^H_{rb}+ \\ 
 (N_t\eta_e\sigma_r^2+N_0) \mathbf{I}_{M}]^{-1}  \mathbf{\Sigma}_{l} $$  and 
 $$ \bar{\mathbf{B}}'_{l}=\mathbf{\Sigma}_{l} - \beta_l \epsilon_{u,p,l} \mathbf{\Sigma}_{l} [\beta_l \epsilon_{u,p,l} \mathbf{\Sigma}_{l}+ \sigma_r^2 \mathbf{\hat{G}}_{rb} \mathbf{\hat{G}}^H_{rb} +(N_t\eta_e\sigma_r^2+N_0) \mathbf{I}_{M}]^{-1}  \mathbf{\Sigma}_{l} ,  $$
 when radar signal information is not known at the BS.}

\subsection{DoA Estimation at the Radar}
The received signal at the radar, in the presence of interference caused due to the pilots transmitted by the users can be expressed as 
\begin{equation}
	\mathbf{z}[n]= h_{rr} \mathbf{A}(\theta) \mathbf{s}[n]+\sum_{k=1}^{K}\sqrt{\epsilon_{u,p,k}}\mathbf{g}_{kr} \psi_k[n]+\sqrt{N}_0 \mathbf{w}_r[n].
\end{equation}
The signal received by the radar can now be reduced to the standard desired signal+noise form, where a standard signal processing technique such as the Multiple Signal Classification~(MUSIC) algorithm can be used to extract the AoA~\cite{6476563}. 
Letting $\mathbf{Z}=[\mathbf{z}[1], \mathbf{z}[2], \ldots, \mathbf{z}[N]]$, we can express it as,
\begin{equation}
	\mathbf{Z}=h_{rr}\mathbf{A}(\theta)\mathbf{S}+\sum_{k=1}^{K}\sqrt{\epsilon_{u,p,k}}\mathbf{g}_{kr}\boldsymbol{\psi}_k[N]+\sqrt{N_0}\mathbf{W}_r,
\end{equation}
with $\mathbf{S}=[\mathbf{s}[1], \mathbf{s}[2], \ldots, \mathbf{s}[N]] \in \mathcal{C}^{N_{t} \times N}$, $\mathbf{W}=[\mathbf{w}[1], \mathbf{w}[2], \ldots, \mathbf{w}[N]] \in \mathcal{C}^{N_{r} \times N}$, and $\boldsymbol{\psi}_k[N]=[\psi_k[1],\ldots,\psi_k[N]] \in \mathcal{C}^{1 \times N}$.
 Consequently, we can write the sample covariance matrix of the received radar signal, $	\hat{\mathbf{R}}_{zz}$, as,
\begin{multline}
	\hat{\mathbf{R}}_{zz}=\mathbf{ZZ}^H=|h_{rr}|^2\mathbf{A}(\theta)\mathbf{SS}^H\mathbf{A}^H(\theta)\\
	+\sum_{k=1}^K\sum_{l=1}^K\sqrt{\epsilon_{u,p,k}}\sqrt{\epsilon_{u,p,l}}\mathbf{g}_{kr}\boldsymbol{\psi}_{k}[N]\boldsymbol{\psi}^H_{l}[N]\mathbf{g}^H_{lr}\\
	+N_0\mathbf{W}_r\mathbf{W}_r^H+2\Re\left\{ h_{rr}\mathbf{A}(\theta)\mathbf{S}\left(\sum_{k=1}^{K}\sqrt{\epsilon_{u,p,k}}\boldsymbol{\psi}^{H}_k[N] \mathbf{g}^{H}_{kr}\right) \right.\\
	 +\left(\sum_{k=1}^{K}\sqrt{\epsilon_{u,p,k}}\mathbf{g}_{kr} \boldsymbol{\psi}_k[N]\right)(\sqrt{N}_{0}\mathbf{W}_{r}^{H})\\\left.+\sqrt{N}_{0} h^{*}_{rr}\mathbf{W}_{r} \mathbf{S}^H\mathbf{A}^H(\theta)\right\},
\end{multline}
 Considering $N=N_t$, we can reduce $\mathbf{SS}^H=\sigma_r^2\mathbf{I}_{N_r}$. Consequently, it is easy to show that the actual covariance matrix of $\mathbf{z}[n]$ takes the form
\begin{equation}
	\mathbf{R}_{zz}=\sigma^2_r |h_{rr}|^{2}\mathbf{A}(\theta)\mathbf{A}^H(\theta)+\left(\sum_{k=1}^{K} \epsilon_{u,p,k} \eta_{rk} +N_{0}\right) \mathbf{I}_{N_{r}}
\end{equation}
Now, $\mathbf{A}(\theta)$ is a rank-1 matrix, and so is $\mathbf{A}(\theta)\mathbf{A}^H(\theta)$, therefore,  $\mathbf{R}_{zz}$ can still be viewed as the sum of a rank-1 matrix and a multiple of the identity matrix. Consequently, the noise subspace of $\mathbf{R}_{zz}$ consists of the eigenvectors corresponding to the $N_r-1$ smallest eigenvalues~(each being equal to $\sum_{k=1}^{K} \epsilon_{u,p,k} \eta_{rk} +N_{0}$). We let the noise subspace of $\mathbf{R}_{zz}$ be represented by the matrix $\mathbf{V}$, and can express the estimate $\hat{\theta}$ of ${\theta}$ as~\cite{6476563},
\begin{equation}
\hat{\theta}=\arg\max_{\phi}\frac{1}{\mathbf{S}^H\mathbf{A}({\phi})^{H} \mathbf{V} \mathbf{V}^{H} \mathbf{A}({\phi})\mathbf{S}} .	
\end{equation}
Since the MUSIC algorithm is intractable for closed form performance analysis, we will evaluate its performance using Monte Carlo simulations in Section~\ref{sec:sims}.

\section{The Uplink Sub-frame}
 In this section, we analyse the performance of the JRC system during the uplink sub-frame. For this purpose, we first evaluate the rates achievable by the communications subsystem via DE analysis~\cite{Papazafeiropoulos2015Deterministic,Zarei2016I/Q}, and then derive the radar rate via the CRLB on the MSE performance of the radar subsystem.

We can write the received signal at the communication BS as 
 \begin{multline}
 \mathbf{y}[n]= \sum_{k=1}^{K} \sqrt{\beta_k\epsilon_{u,s,k}} \mathbf{\hat{h}}_{k} x_k[n]+ \sum_{k=1}^{K} \sqrt{\beta_k\epsilon_{u,s,k}} \mathbf{\tilde{h}}_{k} x_k[n]\\+\hat{\mathbf{G}}_{rb}\mathbf{s}[n]+\tilde{\mathbf{G}}_{rb}\mathbf{s}[n]+\sqrt{N_0} \mathbf{w}_b[n].
 \end{multline} 
We use MMSE combining at the BS, with the matrix $\mathbf{C}\triangleq\mathbf{R}_{yy|\hat{\mathbf{G}}_{rb},\hat{\mathbf{H}}}^{-1}\hat{\mathbf{H}}$ being the combining matrix, such that,
$\mathbf{R}_{yy|\hat{\mathbf{G}}_{rb},\hat{\mathbf{H}}}$ represents the covariance matrix of $\mathbf{y}[n]$ given the availability of the channel estimates $\hat{\mathbf{G}}_{rb}$ and $\hat{\mathbf{H}}$, and can be expressed as,
{\color{black}
\begin{multline}
 	\mathbf{R}_{yy|\hat{\mathbf{G}}_{rb},\hat{\mathbf{H}}}=\sum_{k=1}^{K} \beta_k \epsilon_{u,s,k} \mathbf{\hat{h}}_k\mathbf{\hat{h}}_k^H+\sigma_r^2\hat{\mathbf{G}}_{rb}\hat{\mathbf{G}}_{rb}^H  \\  +\sum_{k=1}^{K} \beta_k \epsilon_{u,s,k}\bar{\mathbf{B}}_{k} +\left(\sigma_r^{2} \eta_{e} +N_{0}\right) \mathbf{I}_{M}.
\end{multline}}
Consequently, the processed signal vector, $\mathbf{r}[n] \in \mathcal{C}^{K \times 1}$, at the BS at the $n${th} instant is given by
$
\mathbf{r}[n]=\mathbf{C}^{H} \mathbf{y}[n]=\hat{\mathbf{H}}^{H} \mathbf{R}_{yy|\hat{\mathbf{G}}_{rb},\hat{\mathbf{H}}}^{-1} \mathbf{y}[n].
$

Now, letting ${r}_{k}[n]$ be the $k^{th}$ component of $\mathbf{r}[n]$, we can write,
\begin{multline}
{r}_{k}[n]=  \sqrt{\beta_k\epsilon_{u,s,k}} \hat{\mathbf{h}}^{H}_{k} \mathbf{R}_{yy|\hat{\mathbf{G}}_{rb},\hat{\mathbf{H}}}^{-1}\mathbf{\hat{h}}_{k} x_k[n]\\+\sum_{\substack{l=1\\l\neq k}}^{K} \sqrt{\beta_l\epsilon_{u,s,l}} \hat{\mathbf{h}}^{H}_{k} \mathbf{R}_{yy|\hat{\mathbf{G}}_{rb},\hat{\mathbf{H}}}^{-1}\mathbf{\hat{h}}_{l} x_l[n]\\+ \sum_{l=1}^{K}   \sqrt{\beta_l\epsilon_{u,s,l}} \hat{\mathbf{h}}^{H}_{k} \mathbf{R}_{yy|\hat{\mathbf{G}}_{rb},\hat{\mathbf{H}}}^{-1}\mathbf{\tilde{h}}_{l} x_l[n] \\+\sum_{i=1}^{N_t}\hat{\mathbf{h}}^{H}_{k} \mathbf{R}_{yy|\hat{\mathbf{G}}_{rb},\hat{\mathbf{H}}}^{-1}\mathbf{\hat{g}}_{rb,i}{s}_i[n]\\ +\hat{\mathbf{h}}^{H}_{k} \mathbf{R}_{yy|\hat{\mathbf{G}}_{rb},\hat{\mathbf{H}}}^{-1}\mathbf{\tilde{G}}_{rb}\mathbf{s}[n]+\sqrt{N_0} \hat{\mathbf{h}}^{H}_{k} \mathbf{R}_{yy|\hat{\mathbf{G}}_{rb},\hat{\mathbf{H}}}^{-1}\mathbf{w}_b[n]. 
\label{eq:r_BS_1}
\end{multline}
Here, the first term corresponds to the desired signal, the second to the cancellable inter user interference, the third to the interference due to the channel estimation errors at the BS, the fourth to the cancellable interference from the radar subsystem, the fifth to the non-cancellable interference from the radar subsystem and the last term to the additive white Gaussian noise.

\begin{theorem}
	\label{thm:uplink_comm}
	The rate achievable by the $k$th user in the uplink of the communication subsystem of a massive MIMO based JRC subsystem can be expressed as 
	\begin{equation}
		R_k=\log_2(1+\gamma_{u,k}),
	\end{equation} 
	where $\gamma_{u,k}$ is the SINR for the $k$th user's signal at the BS, and is given as,
	\begin{equation}
		\gamma_{u,k}=  \frac{\zeta_{s,k}}{\zeta_{I,k}+\zeta_{E,k}+\zeta_{RC,k}+\zeta_{RE,k}+\zeta_{w,k}}.
	\end{equation}
Here $\zeta_{s,k}$ corresponds to the desired signal's power and is given by,
 \begin{equation}
	\zeta_{s,k}= \beta_k\epsilon_{u,s,k} \frac{|\mu_k|^{2}}{|1+\mu_k|^{2}},
	\label{eq:zeta_uxk}
\end{equation}
such that
$\textcolor{black}{\mu_k= \text{Tr}\{\mathbf{T}_{k}(\rho) \mathbf{B}_{k}\}}$, and
{\color{black}
\begin{multline}
	\mathbf{T}_{k}(\rho)=\left(\sum_{\substack{m=1,\\ m\neq k}}^{K} \left( \frac{\beta_m \epsilon_{u,s,m} \mathbf{B}_{m}  }{1+\delta_{k,m}(\rho)}+ \frac{N_{r}\sigma_r^2 (\eta_{I}-\eta_{e}) \mathbf{I}_{M}}{(K-1)(1+\delta_{k,m}(\rho))}  \right) + \right. \ \\ \mathbf{S}+ \rho \mathbf{I}_{M} \Bigg)^{-1},
\end{multline}}
with $\delta_{k,m}(\rho)=\lim_{t \rightarrow \infty} \delta_{k,m}^{(t)}(\rho)$, such that 
{\color{black}
\begin{multline}
	\delta_{k,m}^{(t)}(\rho)= \text{Tr} \Bigg\{\left(\beta_m \epsilon_{u,s,m}\mathbf{B}_{m} + \frac{N_{r} \sigma_r^2 (\eta_{I}-\eta_{e})}{K-1} \mathbf{I}_{M}\right) \\ \times\left(\sum_{l=1, \neq k}^{K} \left(\frac{\beta_l \epsilon_{u,s,l} \mathbf{B}_{l} }{1+\delta_{k,l}^{t-1}(\rho)} + \frac{ N_{r}\sigma_r^2 (\eta_{I}-\eta_{e}) \mathbf{I}_{M}}{(K-1)(1+\delta_{k,l}^{t-1}(\rho))} \right)  + \right.  \\  \mathbf{S} +\rho \mathbf{I}_{M}\Bigg)^{-1} \Bigg\},
\end{multline}}
having initial values $\delta_{k,m}^{(0)}(\rho)=\frac{1}{\rho} \ \forall m  $.

  Similarly, $\zeta_{I,k}$ corresponds to the inter user interference power, and is given as,
  {\color{black}
 \begin{multline}
\zeta_{I,k}=  { \sum_{\substack{l=1, \\ l \neq k}}^{K}} \frac{\beta_l \epsilon_{u,s,l}}{|1+\mu_{k}|^{2}} \Bigg(\mu^{'}_{k,l}+\frac{( \mu^{'}_{k,l})^2}{|1+
	\mu_{k,l}|^{2}}\\- 2 \Re\left\{\frac{( \mu^{'}_{k,l})^{3/2}}{1+ \mu_{k,l}}\right\}\Bigg). 
\label{eq:zeta_IKR}
\end{multline}}
where $\textcolor{black}{\mu_{k,l}=\text{Tr}\{\mathbf{T}_{k,l}(\rho) \mathbf{B}_{l} \}}$, 
{\color{black}
\begin{multline}
\mathbf{T}_{k,l}(\rho)=\left(\sum_{\substack{m=1,\\ m\neq k,l}}^{K} \left( \frac{\beta_m \epsilon_{u,s,m} \mathbf{B}_{m} }{1+\delta_{k,l,m}(\rho)}+ \right.  \right. \\ \left. \left. \frac{ N_{r} \sigma_r^2 (\eta_{I}-\eta_{e}) \mathbf{I}_{M}}{(K-2)(1+\delta_{k,l,m}(\rho))} \right) +  \mathbf{S}+ \rho \mathbf{I}_{M} \right)^{-1},
\label{e6}
\end{multline}}
with $\delta_{k,l,m}(\rho)=\lim_{t \rightarrow \infty} \delta_{k,l,m}^{(t)}(\rho)$, such that
{\color{black}
\begin{multline}
\delta_{k,l,m}^{(t)}(\rho)= \text{Tr} \Bigg \{ \left(\beta_m \epsilon_{u,s,m} \mathbf{B}_{m} +\frac{N_{r}\sigma_r^2 (\eta_{I}-\eta_{e}) }{K-2} \mathbf{I}_{M} \right) \\ \left. \times \left(\sum_{\substack{p=1, \\ p\neq k,l}}^{K} \left(\frac{\beta_p \epsilon_{u,s,p} \mathbf{B}_{p} }{1+\delta_{k,l,p}^{t-1}(\rho)}+  \right. \right. \right.   \\ \left. \left. \left.   \frac{N_{r}\sigma_r^2 (\eta_{I}-\eta_{e}) \mathbf{I}_{M}}{(K-2)(1+\delta_{k,l,p}^{t-1}(\rho))} \right)+\mathbf{S}+ \rho \mathbf{I}_{M}\right)^{-1} \right\},
\end{multline}}
having initial values $\delta_{k,l,m}^{(0)}(\rho)=\frac{1}{\rho} \ \forall m  $,
and  $\textcolor{black}{\mu'_{k,l}=\text{Tr}\{\mathbf{B}_{l}\mathbf{T}^{'}_{k,l}(\rho)\}}$
where $\mathbf{T}^{'}_{k,l}(\rho) \in \mathcal{C}^{M \times M}$ is given by 
{\color{black}
\begin{multline}
\mathbf{T}^{'}_{k,l}(\rho)  = \mathbf{T}_{k,l}(\rho) \mathbf{B}_{k} \mathbf{T}_{k,l}(\rho)  + \mathbf{T}_{k,l}(\rho) \times \\ \sum_{\substack{m=1, \\ m \neq k,l}}^{K} \left( \frac{\beta_m \epsilon_{u,s,m} \mathbf{B}_{m} \delta_{m}^{'}(\rho)}{(1+\delta_{k,l,m}(\rho))^{2}}+ \frac{N_{r} \sigma_r^2 (\eta_{I}-\eta_{e}) \mathbf{I}_{M} \delta_{m}^{'}(\rho)}{(K-2)(1+\delta_{k,l,m}(\rho))^{2}} \right) \times \\ \mathbf{T}_{k,l}(\rho), 
\label{delta_dash} 
\end{multline}} 
and $\boldsymbol{\delta}^{'}(\rho)= [\delta_{1}^{'}(\rho) \ldots\delta_{K}^{'}(\rho)]^{T}$ such that
$
\boldsymbol{\delta}^{'}(\rho)=(\mathbf{I}-\mathbf{J}(\rho))^{-1} \mathbf{v}(\rho),
$
with \textcolor{black}{
$
[\mathbf{J}(\rho)]_{pq}=\frac{1}{(1+\delta_{k,l,p}(\rho))^{2}} \times  \text{Tr}\{(\beta_p \epsilon_{u,s,p} \mathbf{B}_{p} + \frac{N_{r}}{K-2}\sigma_r^2 (\eta_{I}-\eta_{e}) \mathbf{I}_{M}) \mathbf{T}_{k,l}(\rho) \times  (\beta_q \epsilon_{u,s,q} \mathbf{B}_{q} + \frac{N_{r}}{K-2}\sigma_r^2 (\eta_{I}-\eta_{e}) \mathbf{I}_{M}) \mathbf{T}_{k,l}(\rho) \} ,
$ and
$
[\mathbf{v}(\rho)]_{p}=\text{Tr}\{(\beta_p \epsilon_{u,s,p} \mathbf{B}_{p} + \frac{N_{r}}{K-2}\sigma_r^2 (\eta_{I}-\eta_{e}) \mathbf{I}_{M}) \mathbf{T}_{k,l}(\rho) \mathbf{B}_{k} \mathbf{T}_{k,l}(\rho) \}.
$}

The term $\zeta_{E,k}$ corresponds to the interference power due to channel estimation error and is given by
\begin{equation}
\zeta_{E,k}= \sum_{l=1}^{K}   \beta_l\epsilon_{u,s,l} \frac{ \mu^{'}_{k}}{|1+\mu_{k} |^{2}},
\label{eq:zeta_EkF}
\end{equation}
where $\mu^{'}_{k}=\text{Tr}\{\mathbf{T}^{'}_{k}(\rho) \bar{\mathbf{B}}_{l} \} $
and $\mathbf{T}^{'}_{k}(\rho) \in \mathcal{C}^{M \times M}$ is given by
{\color{black} 
\begin{multline}
\mathbf{T}^{'}_{k}(\rho)  = \mathbf{T}_{k}(\rho) \mathbf{B}_k \mathbf{T}_{k}(\rho)  + \mathbf{T}_{k}(\rho) \times \\ \sum_{\substack{m=1,\\ m \neq k}}^{K} \left( \frac{\beta_m \epsilon_{u,s,m} \mathbf{B}_{m} \delta_{m}^{'}(\rho)}{(1+\delta_{k,m}(\rho))^{2}} +\frac{N_{r} \sigma_r^2 (\eta_{I}-\eta_{e}) \mathbf{I}_{M} \delta_{m}^{'}(\rho)} {(K-2)(1+\delta_{k,m}(\rho))^{2}} \right) \times \\ \mathbf{T}_{k}(\rho),  
\label{eq:mu_kdash}
\end{multline} }
with $\mathbf{T}_{k}(\rho) $ and $\boldsymbol{\delta}^{'}(\rho)$ as defined in \eqref{e6} and \eqref{delta_dash} respectively.

$\zeta_{RC,k}$ corresponds to the interference at the $k$th user due to the cancellable component of radar interference and is given by
\begin{multline}
\zeta_{RC,k}=\sum_{i=1}^{N_t} \sigma^{2}_r \frac{1}{|1+ \mu_{k}|^{2}} \Biggl\{\mu^{'}_{k,i}  +  \frac{( \mu^{'}_{k,i})^{3/2}}{|1+\mu_{k,i}|^{2}} \Biggl. \\ \Biggr. -2 \Re \Biggl\{\frac{ \mu^{'}_{k,i}}{1+\mu_{k,i}}\Biggr\}\Biggr\},
\label{eq:zeta_RCkF}
\end{multline}
where $\mu^{'}_{k,i}=\text{Tr}\{\mathbf{T}^{'}_{k,i}(\rho) (\eta_{I}-\eta_{e}) \mathbf{I}_{M} \}  $ , $ \mu_{k,i} = \text{Tr}\{\mathbf{T}_{k,i}(\rho) (\eta_{I}-\eta_{e}) \mathbf{I}_{M} \} $ and $\mathbf{T}^{'}_{k,i}(\rho) \in \mathcal{C}^{M \times M}$ is given by 
{\color{black}
\begin{multline}
\mathbf{T}^{'}_{k,i}(\rho)  = \mathbf{T}_{k,i}(\rho) \mathbf{B}_{k} \mathbf{T}_{k,i}(\rho)  + \mathbf{T}_{k,i}(\rho) \times \\ \sum_{m=1, \neq k}^{K} \left( \frac{\beta_m \epsilon_{u,s,m} \mathbf{B}_{m} \delta_{k,i,m}^{'}(\rho)  }{(1+\delta_{k,i,m}(\rho))^{2}} + \right. \\ \left.    \frac{(N_{r}-1) \sigma_r^2 (\eta_{I}-\eta_{e}) \mathbf{I}_{M} \delta_{k,i,m}^{'}(\rho) }{(K-1)(1+\delta_{k,i,m}(\rho))^{2}}  \right)  \mathbf{T}_{k,i}(\rho),  
\end{multline}} 
and $\mathbf{T}_{k,i}(\rho) $ is given by
{\color{black}
\begin{multline}
\mathbf{T}_{k,i}(\rho)  =\left(\sum_{m=1, \neq k}^{K} \frac{\beta_m \epsilon_{u,s,m}\mathbf{B}_{m}}{1+\delta_{k,i,m}(\rho)} + \right. \\ \left.  \frac{(N_{r}-1) \sigma_r^2 (\eta_{I}-\eta_{e}) \mathbf{I}_{M}}{(K-1)(1+\delta_{k,i,m}(\rho)) }+\mathbf{S}+ \rho \mathbf{I}_{M} \right)^{-1},
\end{multline}}
with $\delta_{k,i,m}(\rho) = \lim_{t \rightarrow \infty} \delta_{k,i,m}^{(t)}(\rho)$ being obtained iteratively from
{\color{black}
\begin{multline}
\delta_{k,i,m}^{(t)}(\rho)= \text{Tr} \left\{ \left( \beta_m \epsilon_{u,s,m} \mathbf{B}_{m}+\frac{(N_{r}-1)\sigma_r^2 (\eta_{I}-\eta_{e}) }{K-1} \mathbf{I}_{M} \right)  \right.\\\left. \times \left(\sum_{\substack{u=1, \\u\neq k}}^{K} \left( \frac{\beta_m \epsilon_{u,s,m} \mathbf{B}_{m}}{1+\delta_{k,i,u}^{t-1}(\rho)}+ \frac{(N_{r}-1) \sigma_r^2 (\eta_{I}-\eta_{e}) \mathbf{I}_{M}}{(K-1)(1+\delta_{k,i,u}^{t-1}(\rho))}  \right)+ \right. \right. \\ \left. \left. \mathbf{S}+ \rho \mathbf{I}_{M} \right)^{-1} \right\},
\end{multline}}
after being initialized as $\delta_{k,i,m}^{(0)}(\rho)=\frac{1}{\rho} \ \forall j  $, and $\boldsymbol{\delta}_{k,i}^{'}(\rho)= [\delta_{k,i,1}^{'}(\rho) \ldots\delta_{k,i,K}^{'}(\rho)]^{T}$ such that
$\boldsymbol{\delta}_{k,i}^{'}(\rho)=(\mathbf{I}_{K}-\mathbf{J}_{k,i}(\rho))^{-1} \mathbf{v}_{k,i}(\rho),
$
with
\textcolor{black}{$
[\mathbf{J}_{k,i}(\rho)]_{pq}=\frac{1}{(1+\delta_{k,i,p}(\rho))^{2}} \times \text{Tr}\{(\beta_p \epsilon_{u,s,p} \mathbf{B}_{p}  + \frac{(N_{r}-1)}{K-1}  \sigma_r^2 (\eta_{I}-\eta_{e}) \mathbf{I}_{M} ) \mathbf{T}_{k,i}(\rho)(\beta_q \epsilon_{u,s,q} \mathbf{B}_{q}  + \frac{(N_{r}-1)}{K-1} \sigma_r^2 (\eta_{I}-\eta_{e}) \mathbf{I}_{M}  )  \mathbf{T}_{k,i}(\rho) \} ,
$ and $
[\mathbf{v}_{k,i}(\rho)]_{p}=\text{Tr}\{(\beta_p \epsilon_{u,s,p}\mathbf{B}_{p} + \frac{(N_{r}-1)}{K-1}  \sigma_r^2 (\eta_{I}-\eta_{e}) \mathbf{I}_{M} ) \mathbf{T}_{k,i}(\rho) \mathbf{B}_{k} \mathbf{T}_{k,i}(\rho) \}.
$}

$\zeta_{RE,k}$ corresponds to the interference at the $k$th user due to error in the estimate of inter-system interference channel between BS and radar and is given by
\begin{equation}
\zeta_{RE,k}= \sum_{i=1}^{N_{t}} \sigma^{2}_r  \frac{b^{2}_{k}\eta_{e}\mu^{'}_{k}}{|1+b^{2}_{k}\mu_{k}|^{2}}.
\label{eq:zeta_rek}
\end{equation}
where  $\mu^{'}_{k}=\text{Tr}\{\mathbf{T}^{'}_{k}(\rho) \}$ and is given by \eqref{eq:mu_kdash}.

$\zeta_{w,k}$ corresponds to the interference at $k$th user due to additive Gaussian noise and is given by
\begin{equation}
\zeta_{w,k}=N_{0}\frac{b^{2}_{k} \mu^{'}_{k}}{|1+b^{2}_{k}\mu_{k}|^{2}}.
\end{equation} 
\end{theorem}
\begin{proof}
	See Appendix~\ref{pr:thm_1}.
\end{proof}
We can observe that the use of MMSE combining, along with the effects of channel hardening offered by the large number of antennas available at the BS effectively cancels the radar generated interference in the uplink. This effect is illustrated in the results obtained in Figs~\ref{f7} and~\ref{f6}. 

We next quantify the performance of the radar subsystem in terms of the achievable radar rate according the notion developed in~\cite{chiriyath2015inner}. We note that the received signal at the radar takes the form 
\begin{equation}
	\mathbf{z}[n]= h_{rr} \mathbf{A}(\theta) \mathbf{s}[n]+\sum_{k=1}^{K}\sqrt{\epsilon_{u,p,k}}\mathbf{g}_{kr} \psi_k[n]+\sqrt{N}_0 \mathbf{w}_r[n].
\end{equation} 
\begin{theorem}
	\label{thm:uplink_radar}
	The radar rate can be expressed as~\cite{chiriyath2015inner},
	\begin{equation}
		R_{\text{radar},u}=\log\left(1+\frac{1}{\text{CRB}(\theta)}\right),
	\end{equation} 
	where $\text{CRB}(\theta)$ corresponds to the Cramer-Rao bound on the MSE of the AoA estimate at the radar, given as,
	\begin{equation}
		CRB(\theta)=\frac{N_0+\sum_{k=1}^{K} \epsilon_{u,p,k} \eta_rk}{2\sigma_r^2 |h_{rr}|^2}\frac{1}{\Re\{ \text{Tr}(\mathbf{\dot{A}}(\theta)\mathbf{\dot{A}}^H(\theta))\}},
	\end{equation}
	with $\mathbf{\dot{A}}(\theta)$ representing the derivative of $\mathbf{A}(\theta)$ wrt $\theta$.
\end{theorem} 
\begin{proof}
See Appendix~\ref{pr:thm_2}.
\end{proof}
We can observe in the expression for the CRLB that the numerator contains the noise term as well as the interferences from all the communications users, limiting the performance of the radar subsystem. We can obtain the rate regions for the overall JRC system during the uplink communication frame by using Theorems~\ref{thm:uplink_comm} and~\ref{thm:uplink_radar}. We next look at the performance of the JRC system during the downlink communication subframe.
\section{The Downlink Sub-frame}
In this section,  we analyse the performance of the JRC system during the downlink sub-frame. Letting $\bar{p}_k[n]$ denote the data symbol to be transmitted to the $k$th user, we can write the data vector to be transmitted by the massive MIMO BS in the downlink as $\mathbf{\bar{p}}[n]=[\bar{p}_1[n],\bar{p}_2[n],\ldots,p_K[n],0,\ldots,0]\in \mathcal{C}^{(K+N_{r}) \times 1}$. We also let $\bar{\mathbf{H}}=[\mathbf{\hat{H}} \, \mathbf{\hat{G}}_{br}] \in \mathcal{C}^{M \times (K+N_{r})}$ be the horizontal concatenation of the estimates of the communication channel, and the radar interference channel. Using this, we can define the precoder matrix at the BS as,
$
	\mathbf{Q}  =(\bar{\mathbf{H}} \bar{\mathbf{H}}^{H} + \alpha \mathbf{I}_{M})^{-1} \hat{\mathbf{H}}^{*},
$
with $\alpha$ being the regularization parameter. This is equivalent to forming a null in the direction of interference channel to the radar. Consequently, the precoded downlink signal transmitted by the BS is expressed as
$
	\mathbf{p}[n]  =\mathbf{Q}\text{diag} (\boldsymbol{\epsilon}_{d,s}) \bar{\mathbf{p}}[n],
$
with  $({\boldsymbol{\epsilon}_{d,s}})=[\sqrt{\epsilon_{d,s,1}},\ldots,\sqrt{\epsilon_{d,s,K}}]^T$, and $\sqrt{\epsilon_{d,s,k}}$ representing the downlink energy allocated to the $k$th user after power control.
We can now write the signal received at the $k$th user as,
\begin{equation}
r_{k}[n]  =\sqrt{\beta_k}\mathbf{\hat{h}}_{k}^{T}  \mathbf{p}[n]+\sqrt{\beta_k}\mathbf{\tilde{h}}_{k}^{T}  \mathbf{p}[n]+\mathbf{g}_{rk}^{T} \mathbf{s}[n]+w_{k}[n]. \\
\end{equation}
This can also be written as
\begin{multline}
r_{k}[n]=\sqrt{\beta_k \epsilon_{d,s,k}} \hat{\mathbf{h}}_{k}^{T}  (\bar{\mathbf{H}} \bar{\mathbf{H}}^{H} + \alpha \mathbf{I}_{M})^{-1} \hat{\mathbf{h}}_{k}^{*} \bar{p}_{k}[n] \\+ \sum_{\substack{m=1, \\ m\neq k}}^{K}\sqrt{\beta_k \epsilon_{d,s,m}} \hat{\mathbf{h}}_{k}^{T}  (\bar{\mathbf{H}} \bar{\mathbf{H}}^{H} +  \alpha \mathbf{I}_{M})^{-1} \hat{\mathbf{h}}_{m}^{*} \bar{p}_{m}[n] \\ +\sqrt{\beta_k}\mathbf{\tilde{h}}_{k}^{T} (\bar{\mathbf{H}} \bar{\mathbf{H}}^{H} + \alpha \mathbf{I}_{M})^{-1} \hat{\mathbf{H}}^{*} \bar{\mathbf{p}}[n]+\mathbf{g}_{rk}^{T} \mathbf{s}[n]+w_{k}[n].
\end{multline}
where the first term indicates the desired signal corresponding to the $k$th user, the second term is the cancellable inter user interference from the data meant for the other users, the third term corresponds to the interference due to the channel estimation error at the BS, the fourth term is due to the interference from the radar subsystem and the last term is due to AWGN.    
\begin{theorem}
	\label{thm:downlink_comm}
		The rate achievable by the $k$th user in the downlink of the communication subsystem of a JRC system can be expressed as 
	\begin{equation}
	R_{d,k}=\log_2(1+\gamma_{d,k}),
	\end{equation} 
	where $\gamma_{d,k}$ is the SINR for the $k$th user's signal at the BS, and is given as,
	\begin{equation}
	\gamma_{d,k}=  \frac{\zeta_{r,k}}{\zeta_{r,I,k}+\zeta_{r,E,k}+\zeta_{r,RC,k}+\zeta_{r,w,k}}.
	\end{equation}
Here $\zeta_{r,k}$ corresponds to the desired signal power and is given by	
\begin{equation}
 		\zeta_{r,k}=\beta_k \epsilon_{d,s,k} \left|\frac{ \mu_{k,\alpha}}{1+ \mu_{k,\alpha}}\right|^{2} 
 		\label{Zeta_{r,k,t}}
 \end{equation}
 with $\textcolor{black}{\mu_{k,\alpha}=\text{Tr}\{\mathbf{T}_{k}(\alpha) \mathbf{B}_{k}} \}$ such that 
 {\color{black}
 \begin{equation}
 \mathbf{T}_{k}(\alpha)=\left(\sum_{\substack{l=1,\\ l \neq k}}^{K} \frac{\mathbf{B}_{l}}{1+\delta_{k,l}(\alpha)}+\sum_{l=K+1}^{K+N_{r}} \frac{(\eta_{I}-\eta_{e}) \mathbf{I}_{M}}{1+\delta_{k,l}(\alpha)}+ \alpha \mathbf{I}_{M} \right)^{-1},
 \end{equation}}
 and $\delta_{k,l}(\alpha)=\lim_{t \rightarrow \infty} \delta_{k,l}^{(t)}(\alpha)$,  which is iteratively computed as 
 {\color{black}
 \begin{multline}
 \delta_{k,l}^{(t)}(\alpha)= \text{Tr} \Biggl\{ \mathbf{B}_{l} \left(\sum_{\substack{m=1,\\ m\neq k}}^{K} \frac{\mathbf{B}_{m}}{1+\delta_{k,m}^{(t-1)}(\alpha)}\right. \\ \left.  + \sum_{m=K+1}^{K+N_{r}} \frac{(\eta_{I}-\eta_{e}) \mathbf{I}_{M}}{1+\delta_{k,m}^{(t-1)}(\alpha)}+ \alpha \mathbf{I}_{M} \right)^{-1} \Biggr\}  \hspace{2 mm} 1\leq l \leq K \\
 \delta_{k,l}^{(t)}(\alpha)= \text{Tr} \Biggl\{(\eta_{I}-\eta_{e})\mathbf{I}_{M} \left(\sum_{\substack{m=1, \\ m\neq k}}^{K} \frac{\mathbf{B}_{m}}{1+\delta_{k,m}^{(t-1)}(\alpha)} + \right. \\ \left.   \sum_{m=K+1}^{K+N_{r}} \frac{(\eta_{I}-\eta_{e}) \mathbf{I}_{M}}{1+\delta_{k,m}^{(t-1)}(\alpha)}+ \alpha \mathbf{I}_{M} \right)^{-1}\Biggr\}  \hspace{2 mm} K+1\leq l \leq K+N_{r}
 \end{multline}}
 with initial values $\delta_{k,l}^{(0)}(\alpha)=\frac{1}{\alpha} \ \forall \ l$. 
The term $\zeta_{r,I,k}$ corresponds to the inter user interference power, and is given as,  
 \begin{multline}
 		\zeta_{r,I,k}=\sum_{\substack{m=1,  m \neq k}}^{K} \beta_k \epsilon_{d,s,m} \frac{1}{|1+ \mu_{k,\alpha}|^{2}} \times \left\{(  \mu^{'}_{k,m,\alpha})+ \right. \\ \left. \frac{( \mu^{'}_{k,m,\alpha})^{2}}{|1+ \mu_{k,\alpha}|^{2}}-2 \Re \left(\frac{( \mu^{'}_{k,m,\alpha})^{3/2}}{1+ \mu_{k,\alpha}}\right) \right\}
 		\label{zeta_{r,I,k,t}},
 \end{multline}
 where $\textcolor{black}{\mu^{'}_{k,m,\alpha}=\text{Tr}\{\mathbf{T}_{k,m}^{'}(\alpha) \mathbf{B}_{m}\}}$, such that
 {\color{black}
\begin{multline}
\mathbf{T}_{k,m}^{'}(\alpha)=\mathbf{T}_{k,m}(\alpha)    \mathbf{B}_{k} \mathbf{T}_{k,m}(\alpha)+\mathbf{T}_{k,m}(\alpha)  \left( \sum_{\substack{l=1,\\ l\neq k,m}}^{K}  \right. \\ \left. \frac{\mathbf{B}_{l}\delta^{'}_{k,m,l}(\alpha)}{1+\delta_{k,m,l}(\alpha)} + \sum_{l=K+1}^{K+N_{r}} \frac{(1-\eta_{e}) \mathbf{I}_{M}\delta^{'}_{k,m,l}(\alpha)}{1+\delta_{k,m,l}(\alpha)}\right)\mathbf{T}_{k,m}(\alpha)
\end{multline}}
with 
{\color{black}
\begin{multline}
\mathbf{T}_{k,m}(\alpha)=\left(\sum_{\substack{l=1,\\ l\neq k,m}}^{K} \frac{\mathbf{B}_{l}}{1+\delta_{k,m,l}(\alpha)} + \right. \\ \left. \sum_{l=K+1}^{K+N_{r}} \frac{(\eta_{I}-\eta_{e}) \mathbf{I}_{M}}{1+\delta_{k,m,l}(\alpha)}+ \alpha \mathbf{I}_{M} \right)^{-1},
\end{multline}}
and $\delta_{k,m,l}(\alpha)=\lim_{t \rightarrow \infty} \delta_{k,m,l}^{(t)}(\alpha)$,  
{\color{black}
\begin{multline}
\delta_{k,m,l}^{(t)}(\alpha)= \text{Tr} \Biggl\{ \mathbf{B}_{l} \left(\sum_{\substack{p=1,\\ p\neq k,m}}^{K} \frac{\mathbf{B}_{p}}{1+\delta_{k,m,p}^{(t-1)}(\alpha)}\right. \\ \left.  + \sum_{p=K+1}^{K+N_{r}} \frac{(\eta_{I}-\eta_{e}) \mathbf{I}_{M}}{1+\delta_{k,m,p}^{(t-1)}(\alpha)}+ \alpha \mathbf{I}_{M} \right)^{-1} \Biggr\}  \hspace{2 mm} 1\leq l \leq K \\
\delta_{k,m,l}^{(t)}(\alpha)= \text{Tr} \Biggl\{(\eta_{I}-\eta_{e})\mathbf{I}_{M} \left(\sum_{\substack{p=1,\\ p\neq k,m}}^{K} \frac{\mathbf{B}_{p}}{1+\delta_{k,m,p}^{(t-1)}(\alpha)}\right. \\ \left.  + \sum_{p=K+1}^{K+N_{r}} \frac{(\eta_{I}-\eta_{e}) \mathbf{I}_{M}}{1+\delta_{k,m,p}^{(t-1)}(\alpha)}+ \alpha \mathbf{I}_{M} \right)^{-1}\Biggr\} \\ K+1\leq l \leq K+N_{r}
\end{multline}}
with initial values $\delta_{k,m,l}^{(0)}(\alpha)=\frac{1}{\alpha} \ \forall \ l$. 

$\boldsymbol{\delta}^{'}_{k,m}(\alpha)= [\delta_{k,m,1}^{'}(\alpha) \ldots\delta_{k,m,K+N_{r}-2}^{'}(\alpha)]^{T}$ such that
$
\boldsymbol{\delta}^{'}_{k,m}(\alpha)=(\mathbf{I}_{K+N_{r}-2}-\mathbf{J}(\rho))^{-1} \mathbf{v}(\rho),
$
$\textcolor{black}{[\mathbf{J}(\alpha)]_{pq}=\frac{\text{Tr}\{\mathbf{B}_{p} \mathbf{T}_{k,m}(\alpha) \mathbf{B}_{q} \mathbf{T}_{k,m}(\alpha) \} }{(1+\delta_{k,m,p}(\alpha))^{2}},}
$ and
$\textcolor{black}{
[\mathbf{v}(\alpha)]_{p}=\text{Tr}\{\mathbf{B}_{p} \mathbf{T}_{k,m}(\alpha) \mathbf{T}_{k,m}(\alpha) \},}
$
when $p,q =1,2,\ldots,K, \neq k,m$  and

 \textcolor{black}{$[\mathbf{J}(\alpha)]_{pq}=\frac{\text{Tr}\{(\eta_{I}-\eta_{e}) \mathbf{I}_{M} \mathbf{T}_{k,m}(\alpha) (\eta_{I}-\eta_{e}) \mathbf{I}_{M} \mathbf{T}_{k,m}(\alpha) \} }{(1+\delta_{k,m,p}(\alpha))^{2}},
	$ and
	$
	[\mathbf{v}(\alpha)]_{p}=\text{Tr}\{(\eta_{I}-\eta_{e}) \mathbf{I}_{M} \mathbf{T}_{k,m}(\alpha) \mathbf{T}_{k,m}(\alpha) \}.
	$
	when $p,q =K+1,\ldots,K+N_{r} $}.

The term $\zeta_{r,E,k}$ corresponds to the interference power due to channel estimation error and is given by
\begin{equation}
\zeta_{r,E,k}=\sum_{l=1 }^{K} \beta_k \epsilon_{d,s,l}  \frac{ \mu^{'}_{l,\alpha}}{|1+\mu_{l,\alpha} |^{2}}
\label{zeta_{r,E,k,t}},
\end{equation}
where $\textcolor{black}{\mu_{l,\alpha}=\text{Tr}\{\mathbf{T}_{l}(\alpha) \mathbf{B}_{l}}\}$ and $ \textcolor{black}{\mu^{'}_{l,\alpha}=\text{Tr}\{\mathbf{T}^{'}_{l}(\alpha)\mathbf{B}_{l}}\}$
such that
{\color{black}
\begin{multline}
\mathbf{T}_{l}^{'}(\alpha)=\mathbf{T}_{l}(\alpha) \bar{\mathbf{B}}_{k} \mathbf{T}_{l}(\alpha)+\mathbf{T}_{l}(\alpha) \left( \sum_{d=1, \neq l}^{K} \frac{\mathbf{B}_{d} \delta^{'}_{l,d}(\alpha)}{1+\delta_{l,d}(\alpha)}\right. \\ \left. +\sum_{d=K+1}^{K+N_{r}} \frac{(\eta_{I}-\eta_{e}) \mathbf{I}_{M}\delta^{'}_{l,d}(\alpha)}{1+\delta_{l,d}(\alpha)}\right)\mathbf{T}_{l}(\alpha),
\end{multline}}
and
{\color{black}
\begin{equation}
\mathbf{T}_{l}(\alpha)=\left(\sum_{\substack{d=1,\\ d \neq l}}^{K} \frac{\mathbf{B}_{d}}{1+\delta_{l,d}(\alpha)}+\sum_{d=K+1}^{K+N_{r}} \frac{(\eta_{I}-\eta_{e}) \mathbf{I}_{M}}{1+\delta_{l,d}(\alpha)}+ \alpha \mathbf{I}_{M} \right)^{-1},
\end{equation}}
where $\delta_{l,d}(\alpha)=\lim_{t \rightarrow \infty} \delta_{l,d}^{(t)}(\alpha)$, 
{\color{black}  
\begin{multline}
\delta_{l,d}^{(t)}(\alpha)= \text{Tr} \Biggl\{ \mathbf{B}_{d} \left(\sum_{\substack{m=1, \\ m \neq l}}^{K} \frac{\mathbf{B}_{m}}{1+\delta_{l,m}^{(t-1)}(\alpha)}\right. \\ \left.  + \sum_{m=K+1}^{K+N_{r}} \frac{(\eta_{I}-\eta_{e}) \mathbf{I}_{M}}{1+\delta_{l,m}^{(t-1)}(\alpha)}+ \alpha \mathbf{I}_{M} \right)^{-1} \Biggr\}  \hspace{2 mm} 1\leq d \leq K \\
\delta_{l,d}^{(t)}(\alpha)= \text{Tr} \Biggl\{(\eta_{I}-\eta_{e})\mathbf{I}_{M} \left(\sum_{\substack{m=1, \\ m \neq l}}^{K} \frac{\mathbf{B}_{m}}{1+\delta_{l,m}^{(t-1)}(\alpha)} + \right. \\ \left.   \sum_{m=K+1}^{K+N_{r}} \frac{(\eta_{I}-\eta_{e}) \mathbf{I}_{M}}{1+\delta_{l,m}^{(t-1)}(\alpha)}+ \alpha \mathbf{I}_{M} \right)^{-1}\Biggr\}  \hspace{5 mm} K+1\leq d \leq K+N_{r}
\end{multline}}
with initial values $\delta_{l,d}^{(0)}(\alpha)=\frac{1}{\alpha} \ \forall \ l$, and 
$
\boldsymbol{\delta}^{'}_{l,d}(\alpha)=(\mathbf{I}_{K+N_{r}-1}-\mathbf{J}(\alpha))^{-1} \mathbf{v}(\alpha),
$

$[\mathbf{J}(\alpha)]_{pq}=\frac{\text{Tr}\{\mathbf{B}_{p} \mathbf{T}_{l}(\alpha) \mathbf{B}_{q}  \mathbf{T}_{l}(\alpha) \} }{(1+\delta_{l,p}(\alpha))^{2}},
$
$
[\mathbf{v}(\alpha)]_{p}=\text{Tr}\{\mathbf{B}_{p}  \mathbf{T}_{k,m}(\alpha) \mathbf{T}_{k,m}(\alpha) \},
$
when $p,q =1,2,\ldots,K, \neq l$ and $[\mathbf{J}(\alpha)]_{pq}=\frac{\text{Tr}\{(\eta_{I}-\eta_{e}) \mathbf{I}_{M} \mathbf{T}_{l}(\alpha) (\eta_{I}-\eta_{e}) \mathbf{I}_{M} \mathbf{T}_{l}(\alpha) \} }{(1+\delta_{l,p}(\alpha))^{2}},
$
$
[\mathbf{v}(\alpha)]_{p}=\text{Tr}\{(\eta_{I}-\eta_{e}) \mathbf{I}_{M}  \mathbf{T}_{k,m}(\alpha) \bar{\mathbf{B}}_{k} \mathbf{T}_{k,m}(\alpha) \},
$
when $p,q =K+1,\ldots,K+N_{r}$ .

 $\zeta_{r,RC,k}$ corresponds to the interference power due to the presence of radar subsystem is,
\begin{equation}
\zeta_{r,RC,k}=\sigma_r^{2} \eta_{r,k}
\end{equation}
$\zeta_{r,W,k}$ corresponds to the interference due to the presence of AWGN and is given by 
\begin{equation}
\zeta_{r,W,k}=N_{0}.
\end{equation}
 \end{theorem}
\begin{proof}
	See~Appendix~\ref{pr:thm_3}.
\end{proof}

Here it is important to note that the RZF beamforming by the BS only nulls the interference from the BS to the radar and does not affect the interference from the radar to the users. The radar to user interference, denoted by $\zeta_{r,RC,k}$ in the SINR expression remains unmitigated. 

We next analyse the performance of radar subsystem in the downlink subframe. We note that the received signal at the radar is given by 
\begin{multline}
	\mathbf{z}[n]= h_{rr} \mathbf{A}(\theta) \mathbf{s}[n]+\mathbf{\hat{G}}_{br}^{H}\mathbf{Q}\text{diag} (\sqrt{\boldsymbol{\epsilon}_{d,s}}) \bar{\mathbf{p}}[n]\\+\mathbf{\tilde{G}}_{br}^{H}\mathbf{Q}\text{diag} (\sqrt{\boldsymbol{\epsilon}_{d,s}}) \bar{\mathbf{p}}[n] +\sqrt{N}_0 \mathbf{w}_r[n].
\end{multline}
\begin{theorem} \label{thm:downlink_radar}
The radar rate can be expressed as,
\begin{equation}
	R_{\text{radar},d}=\log\left(1+\frac{1}{\text{CRB}(\theta)}\right),
\end{equation} 
where 
\begin{equation}
\text{CRB}(\theta)=\frac{	\sigma_{wr,d}^2}{2\sigma_r^2 |h_{rr}|^2}\frac{1}{\Re\{ \text{Tr}(\mathbf{\dot{A}}(\theta)\mathbf{\dot{A}}^H(\theta))\}},
\end{equation}
	with
\begin{equation}
\sigma_{wr,d}^2 =N_0+\frac{ \mu^{'}_{i,\alpha}}{|1+ \mu_{\mathbf{\hat{g}}_{br,m}}|^{2}}+ \mu^{'}_{\alpha}
\end{equation}
such that $\textcolor{black}{\mu^{'}_{\alpha}=\text{Tr}\left\{\mathbf{T}^{'}(\alpha) \left( \sum_{i=1}^{K}\epsilon_{d,s,i} \mathbf{B}_{i} \right)\right\}}$,
{\color{black}
\begin{multline}
	\mathbf{T}^{'}(\alpha)=\mathbf{T}(\alpha)(\eta_{e} \mathbf{I}_{M}) \mathbf{T}(\alpha)+\mathbf{T}(\alpha) \left( \sum_{l=1}^{K} \frac{\mathbf{B}_{l}\delta^{'}_{l}(\alpha)}{1+\delta_{l}(\alpha)}\right. \\ \left. +\sum_{l=K+1}^{K+N_{r}} \frac{(\eta_{I}-\eta_{e}) \mathbf{I}_{M}\delta^{'}_{l}(\alpha)}{1+\delta_{l}(\alpha)}\right)\mathbf{T}(\alpha)
\end{multline}}
with 
{\color{black}
\begin{equation}
	\mathbf{T}(\alpha)=\left(\sum_{l=1}^{K} \frac{\mathbf{B}_{l}}{1+\delta_{l}(\alpha)}+\sum_{l=K+1}^{K+N_{r}} \frac{(\eta_{I}-\eta_{e}) \mathbf{I}_{M}}{1+\delta_{l}(\alpha)}+ \alpha \mathbf{I}_{M} \right)^{-1},
\end{equation}}
and $\delta_{l}(\alpha)=\lim_{t \rightarrow \infty} \delta_{l}^{(t)}(\alpha)$,  
{\color{black}
\begin{multline}
	\delta_{l}^{(t)}(\alpha)= \text{Tr} \Biggl\{ \mathbf{B}_{l} \left(\sum_{p=1}^{K} \frac{\mathbf{B}_{p}}{1+\delta_{p}^{(t-1)}(\alpha)}\right. \\ \left.  + \sum_{p=K+1}^{K+N_{r}} \frac{(\eta_{I}-\eta_{e}) \mathbf{I}_{M}}{1+\delta_{p}^{(t-1)}(\alpha)}+ \alpha \mathbf{I}_{M} \right)^{-1} \Biggr\}  \hspace{2 mm} 1\leq l \leq K \\
	\delta_{l}^{(t)}(\alpha)= \text{Tr} \Biggl\{(\eta_{I}-\eta_{e})\mathbf{I}_{M} \left(\sum_{\substack{p=1,\\ p \neq k,m}}^{K} \frac{\mathbf{B}_{p}}{1+\delta_{p}^{(t-1)}(\alpha)}\right. \\ \left.  + \sum_{m=K+1}^{K+N_{r}} \frac{(\eta_{I}-\eta_{e}) \mathbf{I}_{M}}{1+\delta_{p}^{(t-1)}(\alpha)}+ \alpha \mathbf{I}_{M} \right)^{-1}\Biggr\}  \hspace{5 mm} K+1\leq l \leq K+N_{r}
\end{multline}}
such that $\delta_{l}^{(0)}(\alpha)=\frac{1}{\alpha} \ \forall \ l$,
$\boldsymbol{\delta}^{'}(\alpha)= [\delta_{1}^{'}(\alpha) \ldots\delta_{K+N_{r}-2}^{'}(\alpha)]^{T}$,
$
	\boldsymbol{\delta}^{'}(\alpha)=(\mathbf{I}_{K+N_{r}}-\mathbf{J}(\rho))^{-1} \mathbf{v}(\rho),
$
with
$
	[\mathbf{J}(\alpha)]_{pq}=\frac{\text{Tr}\{\mathbf{B}_{p} \mathbf{T}(\alpha) \mathbf{B}_{q} \mathbf{T}(\alpha) \} }{(1+\delta_{p}(\alpha))^{2}},
	[\mathbf{v}(\alpha)]_{p}=\text{Tr}\{\mathbf{B}_{p} \mathbf{T}(\alpha) \eta_{e} \mathbf{I}_{M} \mathbf{T}(\alpha) \},
$
when $p,q =1,2,\ldots,K$ and $
[\mathbf{J}(\alpha)]_{pq}=\frac{\text{Tr}\{(\eta_{I}-\eta_{e}) \mathbf{I}_{M} \mathbf{T}(\alpha)(\eta_{I}-\eta_{e}) \mathbf{I}_{M}\mathbf{T}(\alpha) \} }{(1+\delta_{p}(\alpha))^{2}},
[\mathbf{v}(\alpha)]_{p}=\text{Tr}\{(\eta_{I}-\eta_{e}) \mathbf{I}_{M} \mathbf{T}(\alpha) \eta_{e} \mathbf{I}_{M} \mathbf{T}(\alpha) \},
$ when $p,q =K+1,\ldots,K+N_{r}$.
Also,
{\color{black}
\begin{equation}
	\mu^{'}_{i,\alpha}=\text{Tr}\left\{\mathbf{T}^{'}_{i}(\alpha) \left(\sum_{i=1}^{K}\epsilon_{d,s,i} \mathbf{B}_{i}\right)\right\}
\end{equation}}
where,
{\color{black}
\begin{multline}
\mathbf{T}^{'}_{i}(\alpha)=\mathbf{T}_{\mathbf{\hat{g}}_{br,m}}(\alpha)(\eta_{I}-\eta_{e}) \mathbf{I}_{M} \mathbf{T}_{\mathbf{\hat{g}}_{br,m}}(\alpha) + \mathbf{T}_{\mathbf{\hat{g}}_{br,m}}(\alpha) \times\\ \left( \sum_{l=1}^{K} \frac{\mathbf{B}_{l}\delta^{'}_{l}(\alpha)}{1+\delta_{l}(\alpha)}\right.  \left. +\sum_{l=K+1}^{K+N_{r}} \frac{(\eta_{I}-\eta_{e}) \mathbf{I}_{M}\delta^{'}_{l}(\alpha)}{1+\delta_{l}(\alpha)}\right)\mathbf{T}_{\mathbf{\hat{g}}_{br,m}}(\alpha)
\end{multline}}
with 
{\color{black}
\begin{multline}
\mathbf{T}_{\mathbf{\hat{g}}_{br,m}}(\alpha)\\=\left(\sum_{l=1}^{K} \frac{\mathbf{B}_{l}}{1+\delta_{l}(\alpha)}+\sum_{\substack{l=K+1, \\ l \neq m}}^{K+N_{r}} \frac{(\eta_{I}-\eta_{e}) \mathbf{I}_{M}}{1+\delta_{l}(\alpha)}+ \alpha \mathbf{I}_{M} \right)^{-1},
\end{multline}}
and
$\mu_{\mathbf{\hat{g}}_{br,m}}=\text{Tr}\{\mathbf{T}_{\mathbf{\hat{g}}_{br,m}}(\alpha) (\eta_{I}-\eta_{e}) \mathbf{I}_{M}\}.$
\end{theorem}
\begin{proof}
	See Appendix~\ref{pr:thm_4}.
\end{proof}
Again, the rate regions quantifying the performance of the proposed JRC system can be obtained by using Theorems~\ref{thm:downlink_comm} and~\ref{thm:downlink_radar} in conjunction. In the next section, we present simulation results to better visualize the ideas presented by these theorems. 
\section{Simulation Results} \label{sec:sims}
\textcolor{black}{In this section we validate our derived results using Monte-Carlo simulations and prescribe parameter values for optimized system operation. 
 Here, the communication subsystem consists of single cell massive MIMO system having $M$=128 antennas (with an inter antenna spacing equal to half the carrier wavelength) at a BS that is located at the cell centre serving $K $ = 8 users distributed uniformly across the cell, and operating at a carrier frequency $f_c$ = 3~GHz. The complex basedband equivalent communications and radar signals are assumed to have a bandwidth of $20$~MHz. The channel covariance matrix is $\boldsymbol{\Sigma}_k$ is assumed to be the identity matrix $(\mathbf{I}_M)$. For simplicity, we assume the cell to be circular, and having a radius 100~m. The communication frame consists of $1024$ channel uses with the first $K=8$ channel uses dedicated for training, and the remaining divided equally for uplink and downlink data transmission. The performance of the communication subsystem, unless stated otherwise is measured in terms of the average per user achievable rate. For the purpose of these experiments, we consider the interference channels to be known at the BS with a $10\%$ error (i.e. $\eta_{e}= 0.1$) and the variance of AWGN to be unity~(i.e. $N_0= 1$). The large scale fading path loss in wireless channels is modelled by using the simple path loss model~\cite{Jakes} with the path loss exponent being 3.6. We also apply channel inversion based power control in both uplink and downlink.}

On the other hand, the radar subsystem consists of a MIMO radar with $N_t=N_r=8$, transmit and receive antennas. Both the radar subsystem and the target are assumed to be located randomly within the cell and their locations follow the same distribution as the users. Unless stated otherwise, the performance of the radar is quantified in terms of the radar rate derived in the previous sections. All the performance metrics presented in this section are generated by averaging over 10,000 realizations of the system.

\subsection{CSI Acqusition}
	\begin{figure} [t]
			\centering
			\includegraphics[width= 0.9\linewidth]{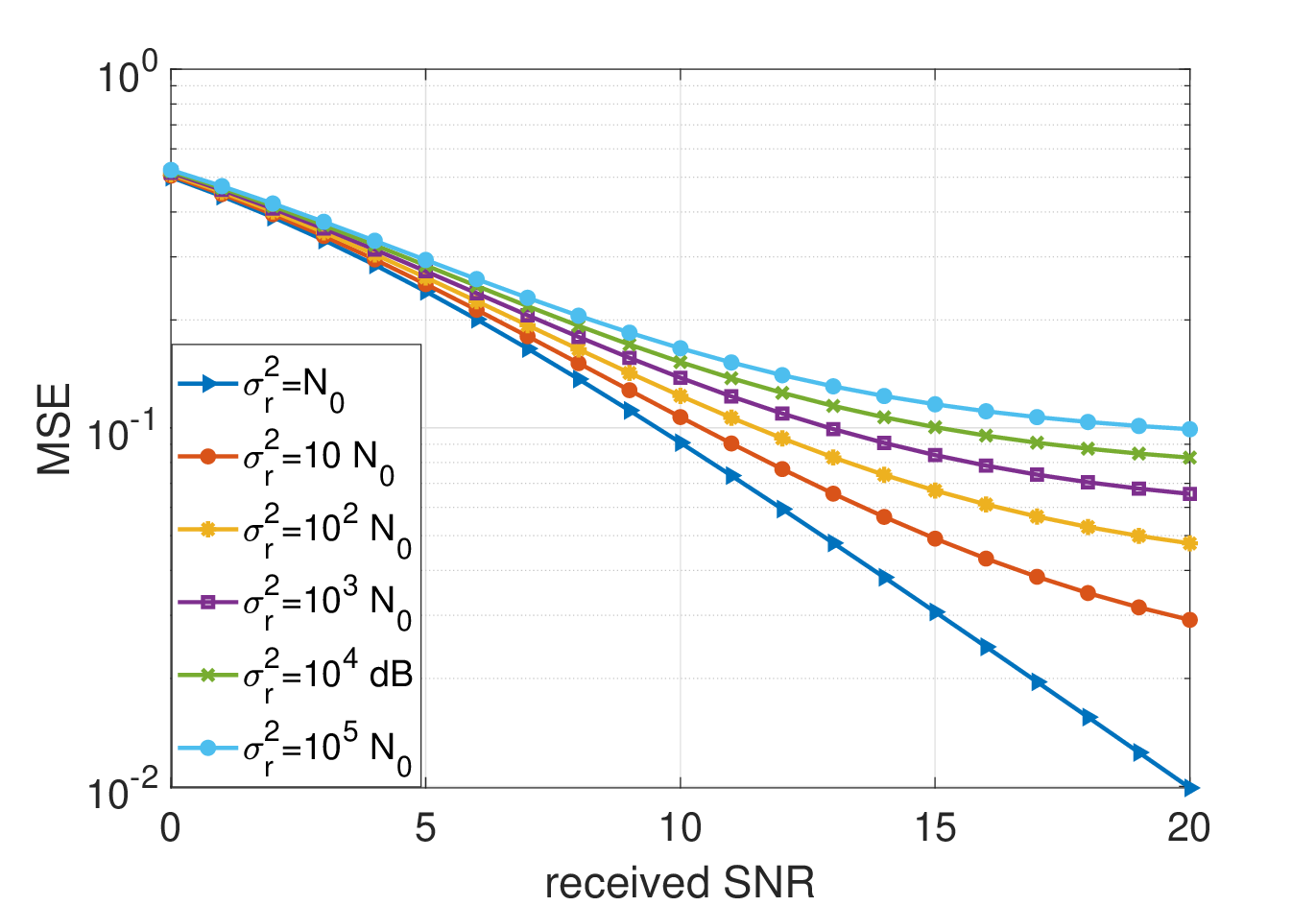}
			\caption{Channel estimation MSE as a function of the received pilot SNR for different values of radar transmit power.}
			\label{f8}
		\end{figure}
	
	 Fig.~$\ref{f8}$ plots the MSE of uplink channel estimate as a function of the received pilot SNR for different values of radar transmit power, $\sigma_r^2$. We observe that a higher radar transmit power does result in a saturation of the channel estimation MSE, however, the impact of radar interference is minimal when both the communication system and the radar are operating at SNRs around 10~dB, which results in a fair channel estimation MSE. 
	\begin{figure} [t]
	\centering
	\includegraphics[width= 0.9\linewidth]{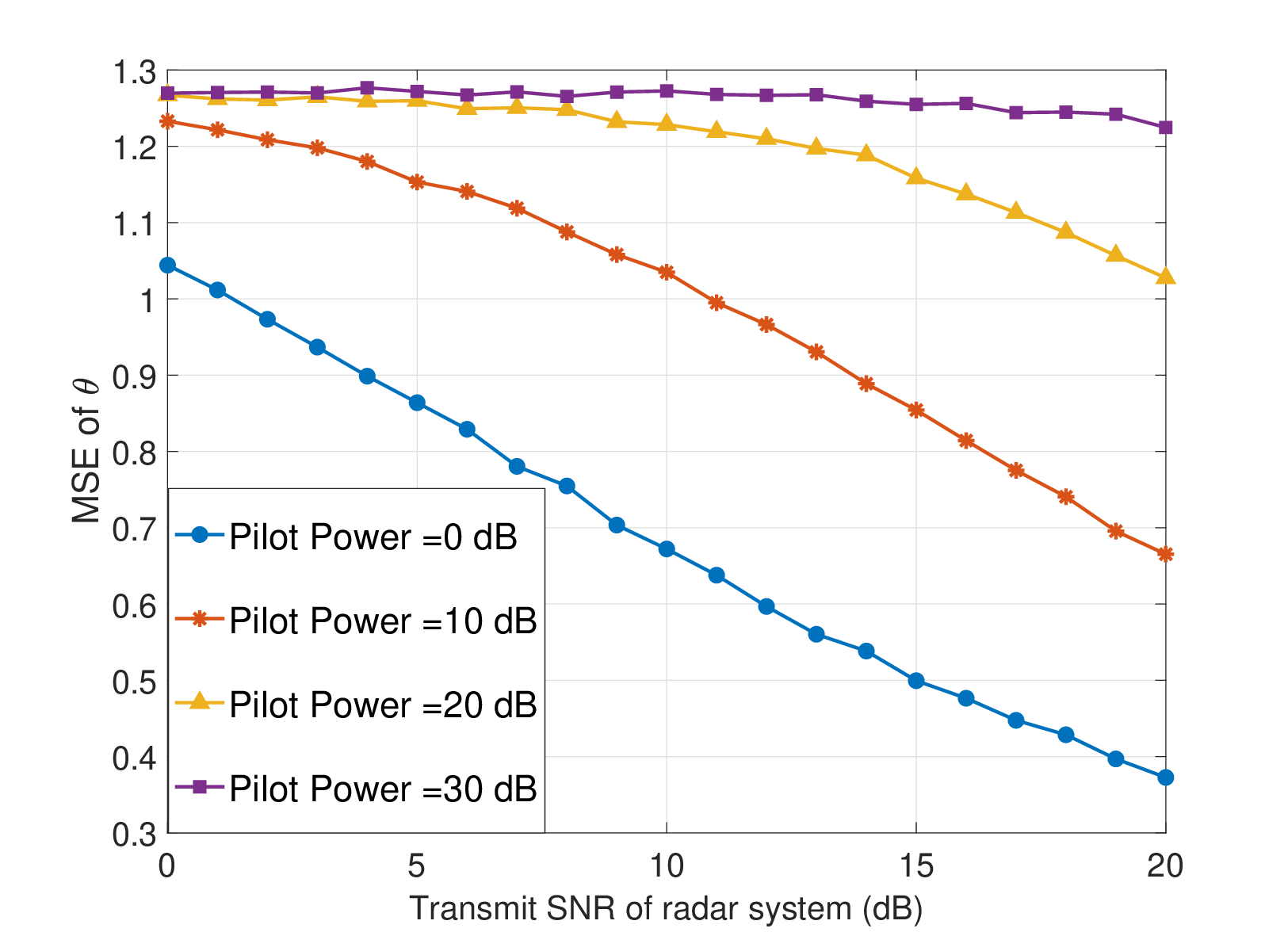}
	\caption{ Radar AoA MSE versus radar SNR for different values of transmit pilot power.}
	\label{f3}
\end{figure}

 In Fig.~$\ref{f3}$ we plot the mean square estimation error of the AoA at the radar as a function of radar received SNR for three different values of pilot powers transmitted by the communication subsystem. \textcolor{black}{We observe that that the radar MSE performance degrades as the pilot power of the communication subsystem, and hence the interference to the radar subsystem increases.}

\textcolor{black}{\subsection{Validation of Asymptotic Approximations}
\begin{figure} [t]
	\centering
	\includegraphics[width= 0.9\linewidth]{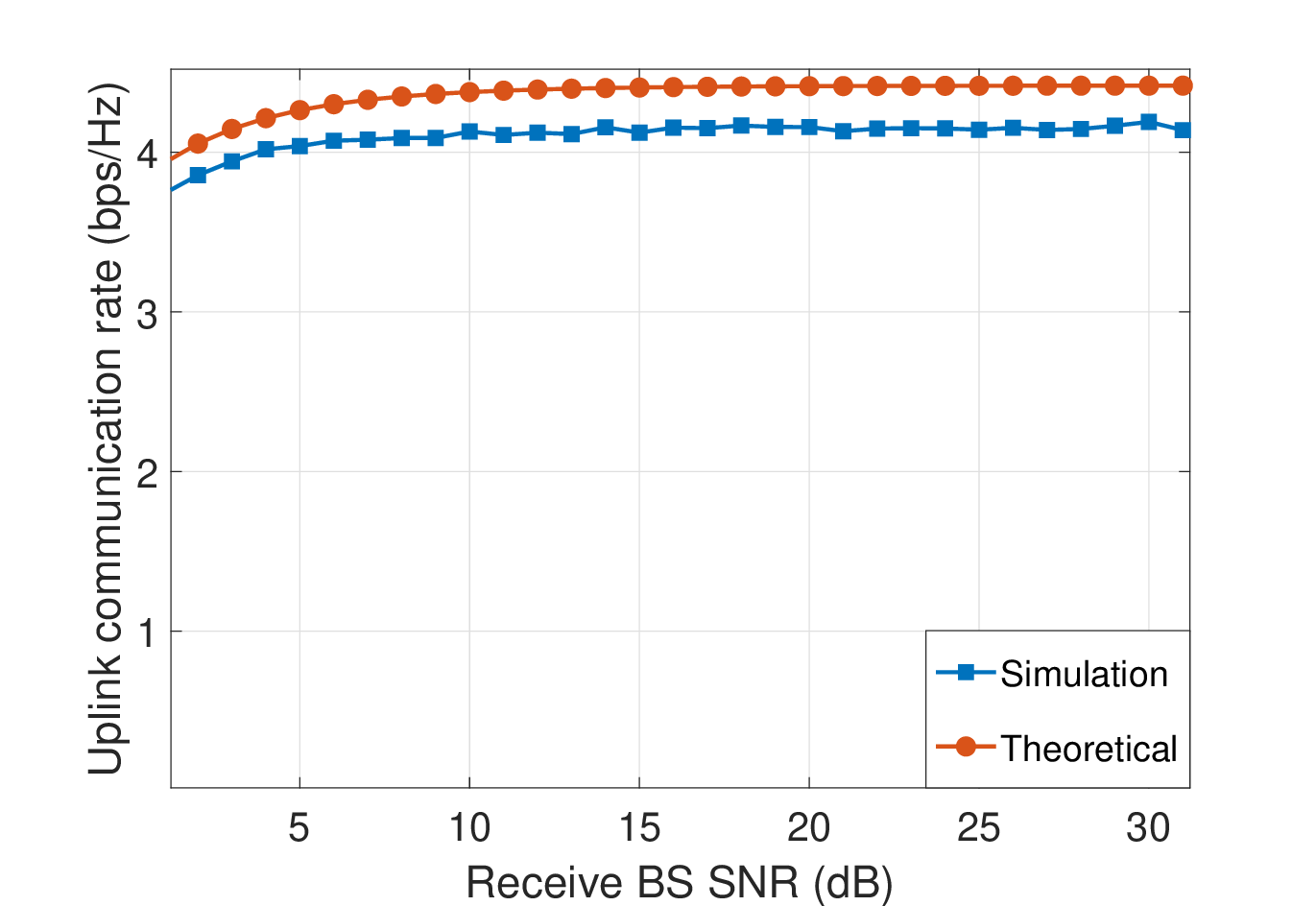}
	\caption{ Average uplink communication rate from both theoretical and simulation analysis versus received data SNR for $\sigma_r^2=N_0$.}
	\label{f55}
\end{figure}
\begin{figure} [t]
	\centering
	\includegraphics[width= 0.9\linewidth]{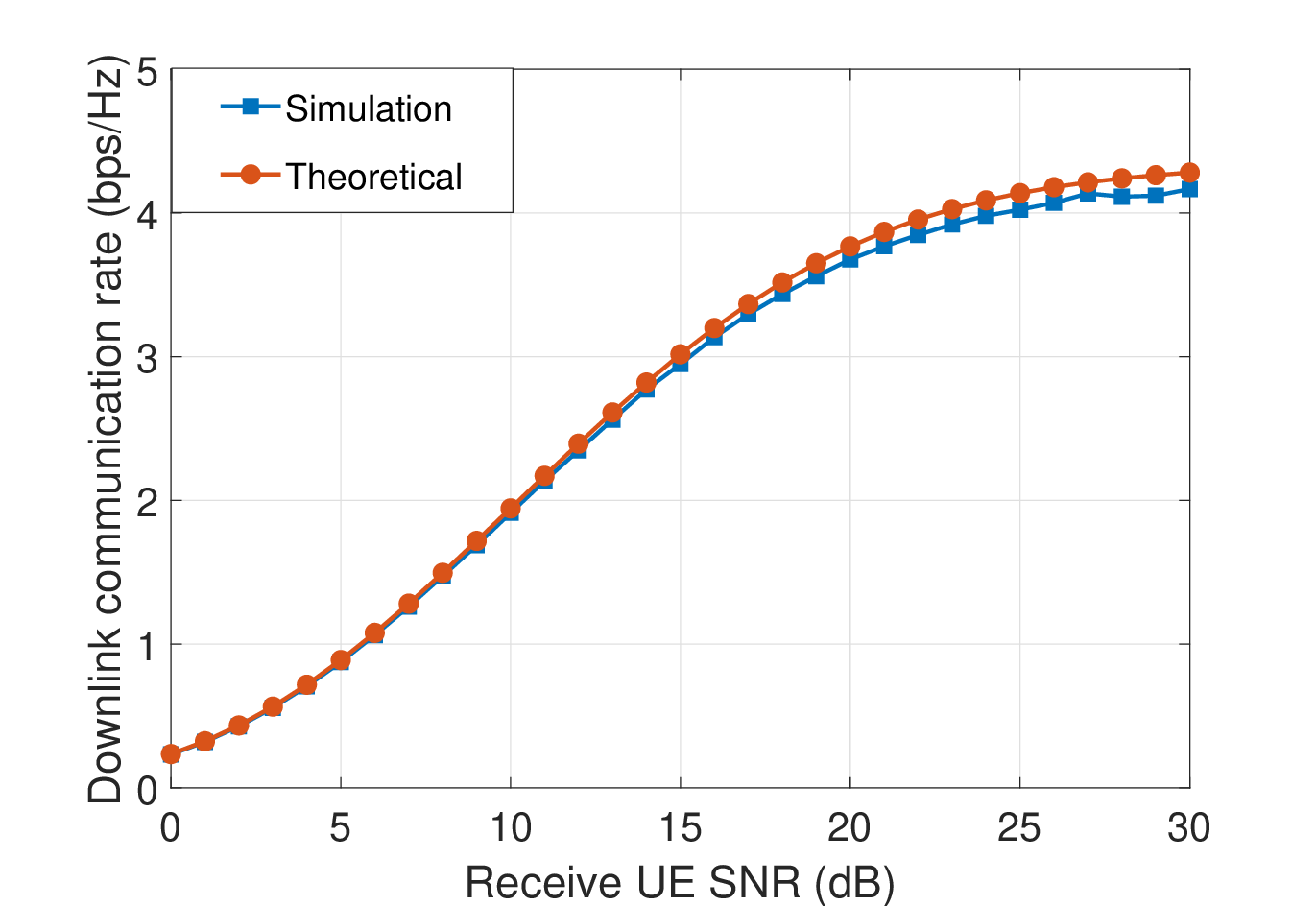}
	\caption{Average communication rate in downlink as a function of the received data SNR from both theoretical and simulation analysis for $\sigma_r^2=N_0$ radar transmitted power.}
	\label{f77}	
\end{figure}
 In Figs.~$\ref{f55}$ and $\ref{f77}$ we respectively plot the achievable uplink and downlink rates of the communication subsystem as functions of the received SNR with $\sigma_r^2=1$. The theoretical values used in these plots are obtained using Theorems~\ref{thm:uplink_comm} and~\ref{thm:downlink_comm} for the uplink and downlink cases respectively. We observe that the simulated results match closely with our derived results, allowing us to use the former for further analysis.}

\subsection{Uplink Data Transmission}
 	\begin{figure} [t]
 	\centering
 	\includegraphics[width= 0.9\linewidth]{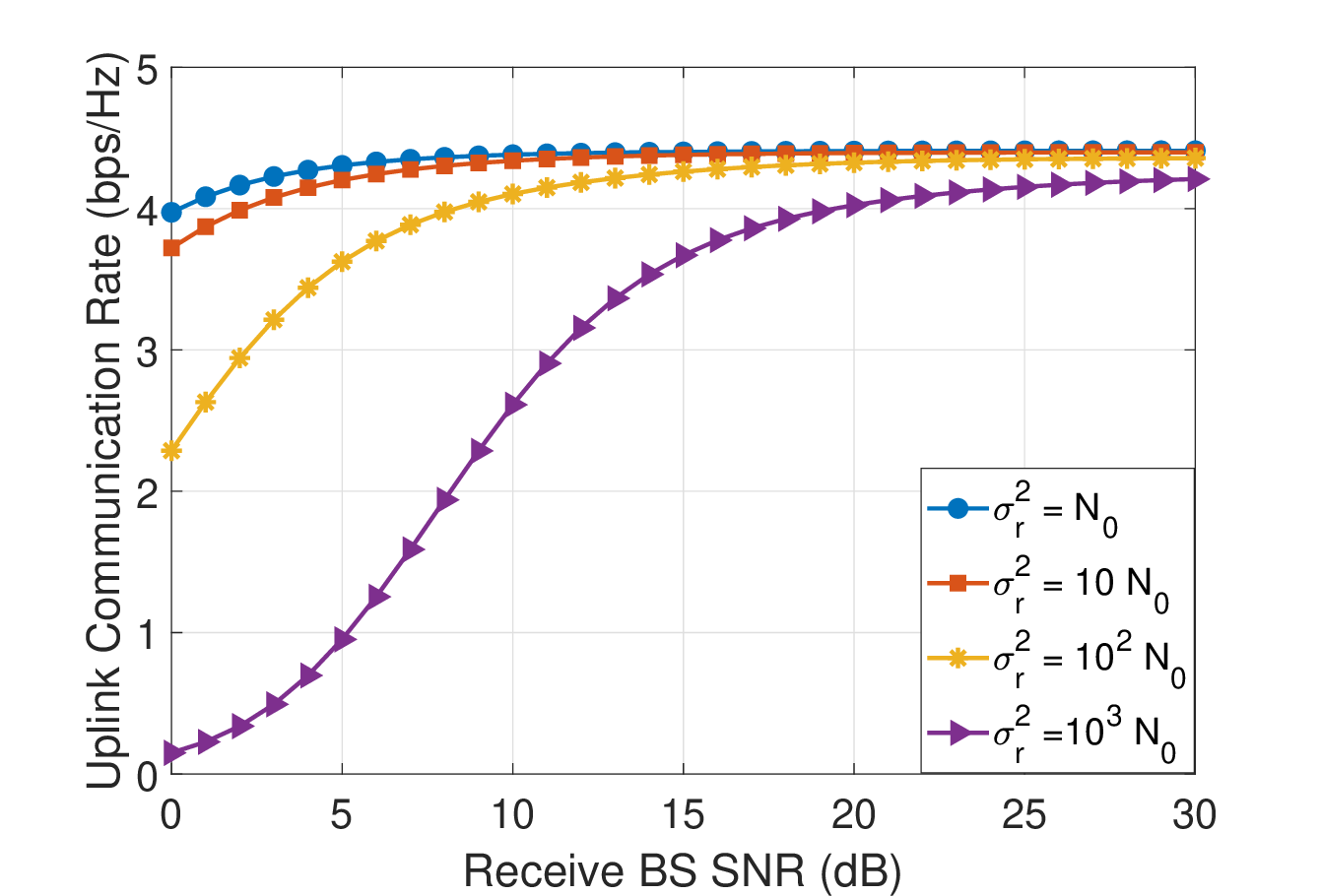}
 	\caption{ Average uplink communication rate versus received data SNR for different values of radar transmit power.}
 	\label{f5}
 \end{figure}
 \textcolor{black}{Fig.~$\ref{f5}$ plots the achievable communication rate in the uplink subframe as a function of received BS SNR for different levels of interference caused by radar subsystem. For higher values of received SNR at BS (i.e. $ >$ 30 dB) we observe that even a radar SNR of 30~dB results in a negligible loss in the average achievable per user rate for the communication system. This is because of the effective cancellation of the radar generated interference at the BS, as postulated in Theorem~\ref{thm:uplink_comm}. Conversely, for low values of BS receive SNR (i.e. $ <$ 15 dB) the loss in the average achievable per user rate is significant for the higher values of radar transmit SNR (i.e. $>$ 20 dB), and therefore, it is safe to conclude that under this SNR regime, the system performance is limited by the radar SNR.}
	\begin{figure} [t]
	\centering
	\includegraphics[width= 0.9\linewidth]{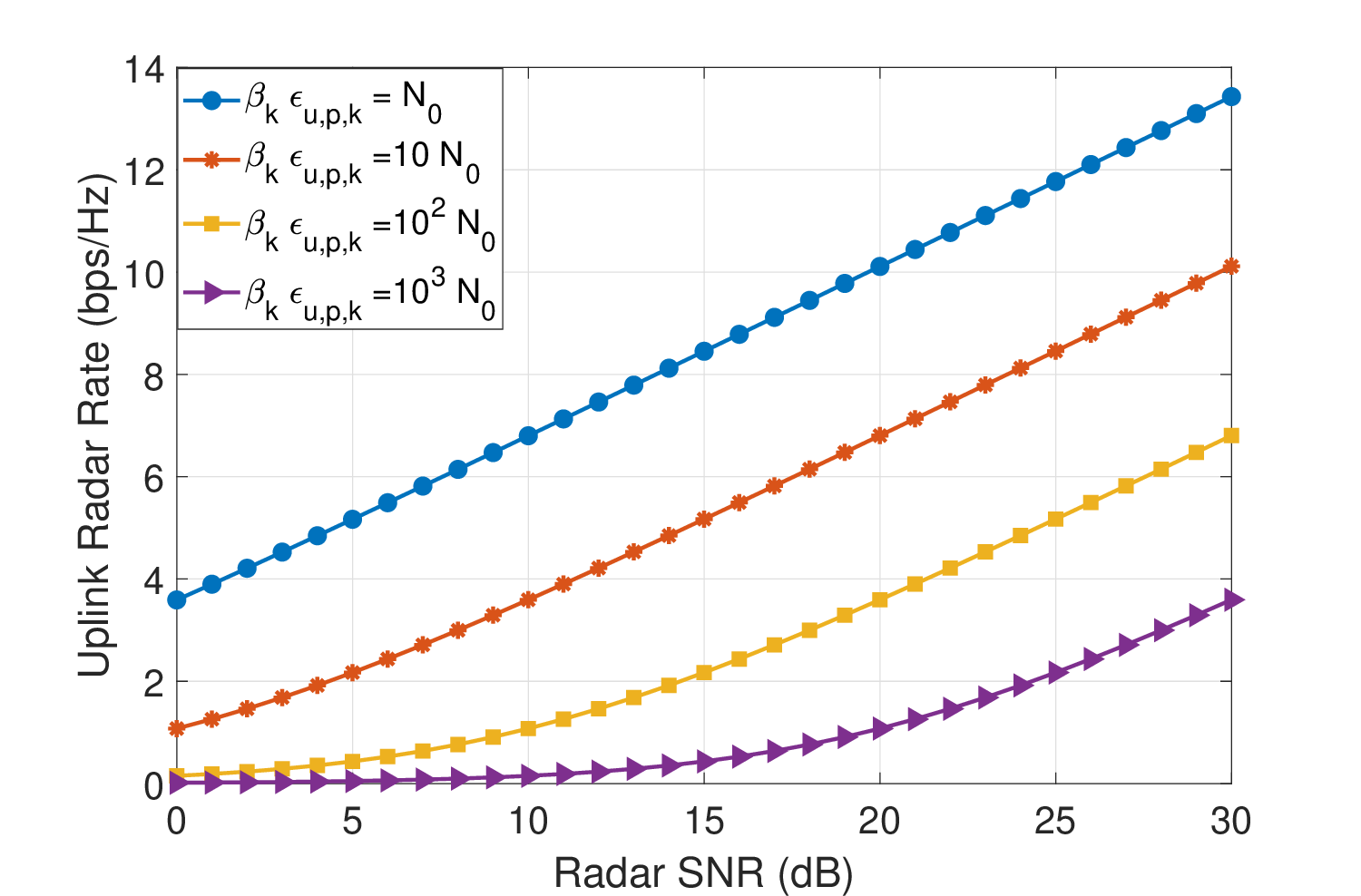}
	\caption{Average uplink radar rate as a function of the radar SNR for different values of uplink data power.}
	\label{f4}
\end{figure}

Fig.~$\ref{f4}$ illustrates the achievable radar rate in the uplink communication subframe as a function of the received SNR for different levels of interference caused by the communication subsystem. We observe a loss of about 3 bits per channel use when the communication subsystem is operating at a received SNR of 10~dB. This result is in line with Theorem~\ref{thm:uplink_radar}, where the radar is shown to face unmitigated interference.
	\begin{figure} [t]
	\centering
	\includegraphics[width= 0.9\linewidth]{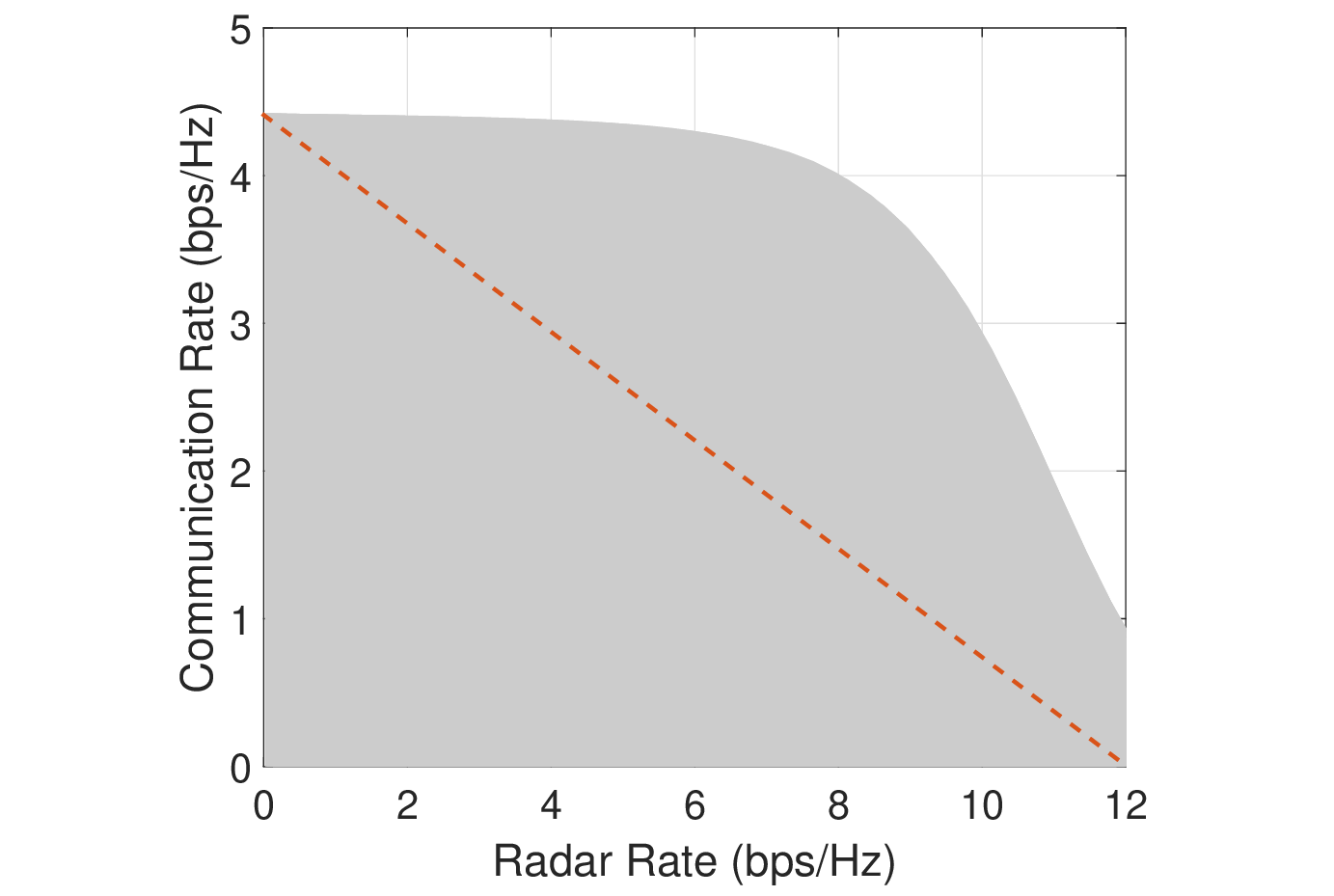}
	\caption{Consolidated Rate Region during the uplink subframe.}
	\label{f9}
\end{figure}

Fig.~$\ref{f9}$ plots the rate region of the JRC system in the uplink communication sub-frame. \textcolor{black}{We observe that the rates achievable by the two subsystems can be traded off with each other, with the dotted line representing the case when the system operates in the time division duplexed mode. We however, observe that the rate region is significantly convex, indicating that the two systems can coexist in a fully shared spectrum with marginal performance losses.} 
	\begin{figure} [t]
	\centering
	\includegraphics[width= 0.9\linewidth]{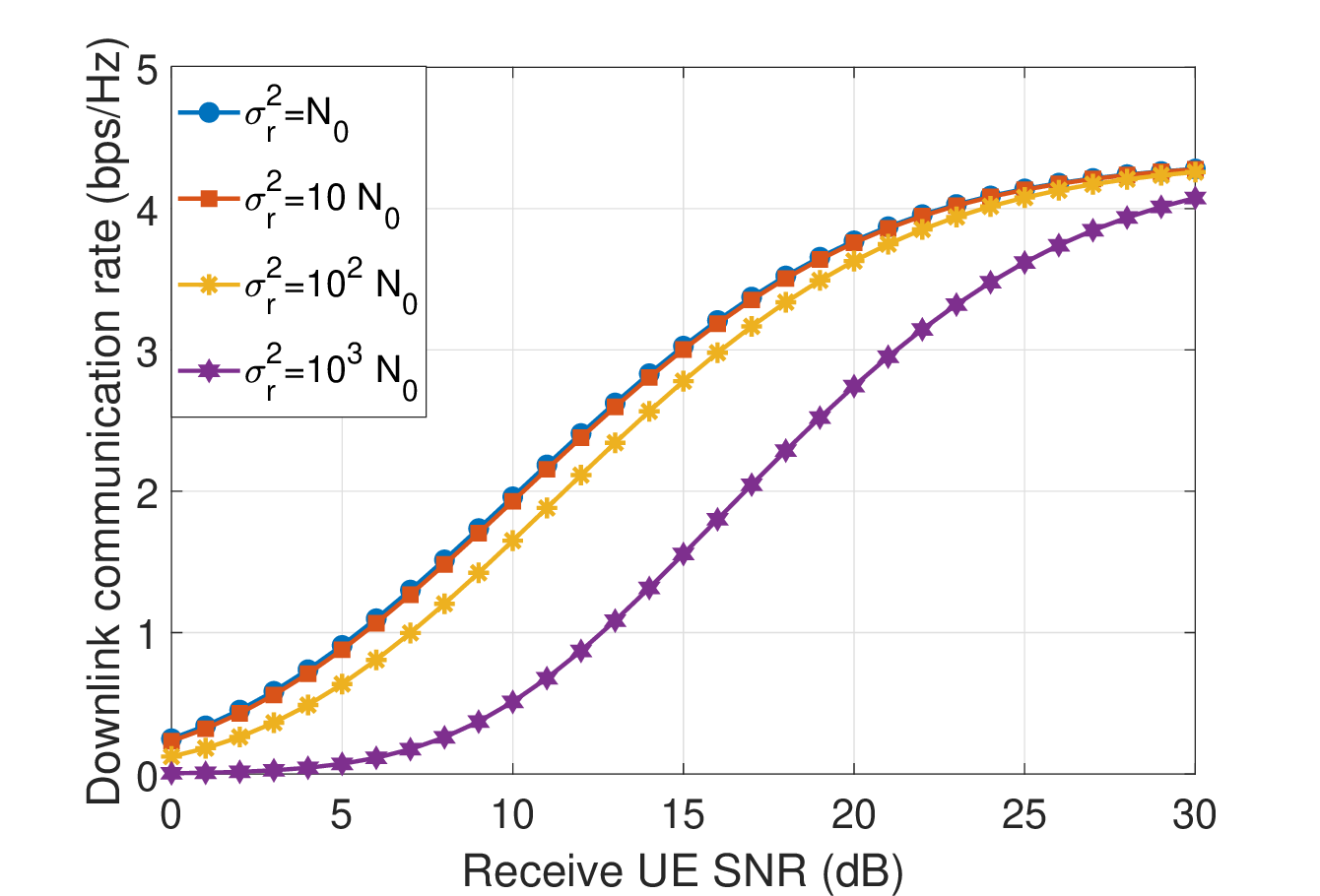}
	\caption{Average communication rate in downlink as a function of the received data SNR for different values of radar transmitted power.}
	\label{f7}	
    \end{figure}
\subsection{Downlink Data Transmission}

Fig.~$\ref{f7}$ illustrates the achievable communication rate in the downlink communication subframe as a function of receive BS SNR for different levels of interference caused by radar subsystem. We observe that in line with the uplink case, the performance degradation is minimal for higher receive UE SNR and the loss is significant for lower values of receive UE SNR. 
	\begin{figure} [t]
	\centering
	\includegraphics[width= 0.9\linewidth]{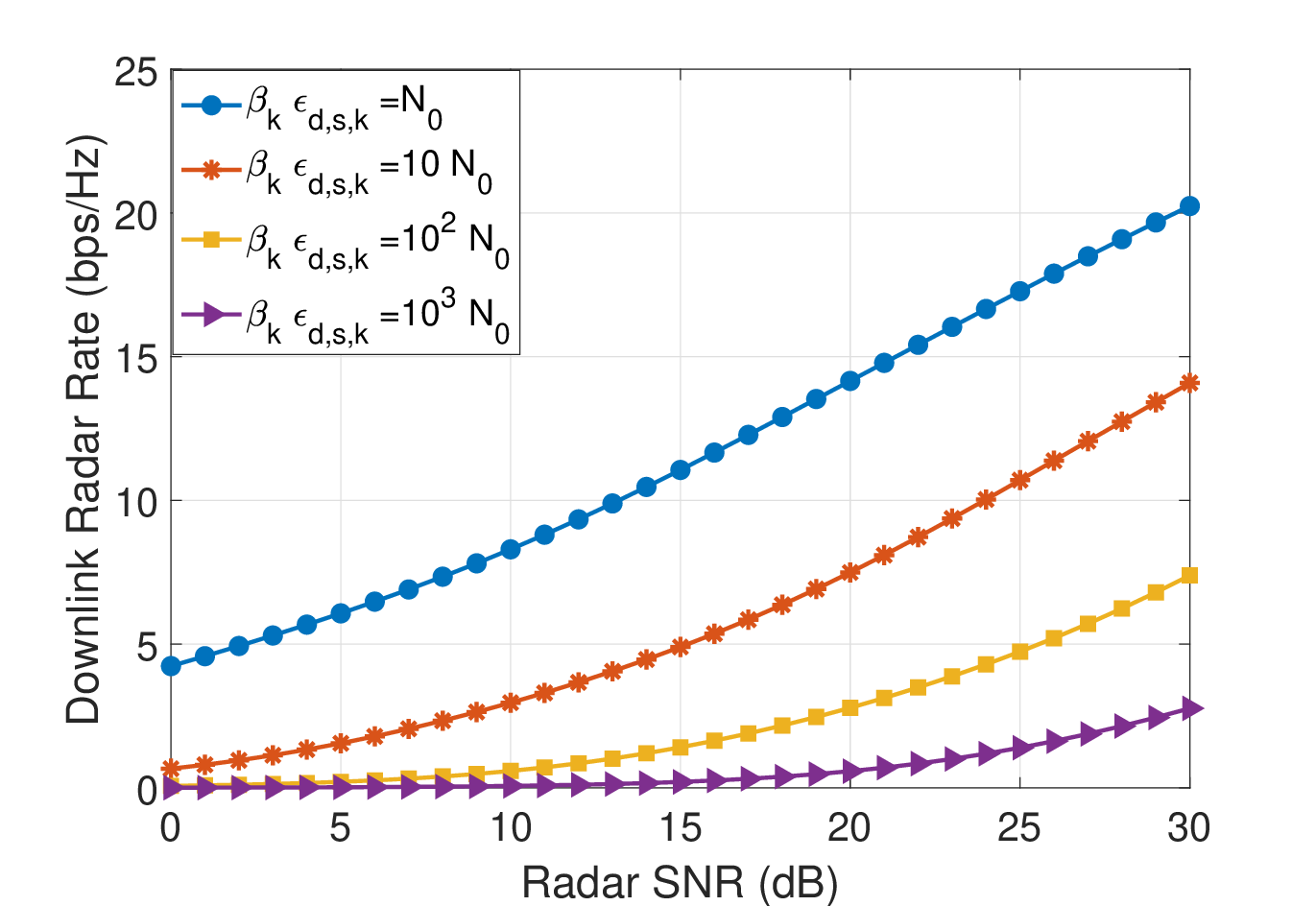}
	\caption{Average radar rate in the downlink as a function of the radar SNR for different values of downlink data power.}
	\label{f6}
    \end{figure}

Fig.~$\ref{f6}$ plots  the achievable radar rate in the downlink communication subframe as a function of the received SNR for different levels of interference by the communication subsystem. It is clearly visible that the radar rate achievable during the downlink sub-frame is better than that achievable during the uplink subframe, indicating the efficacy of the null being formed in the direction of the radar by the massive MIMO BS, as stated in Theorem~\ref{thm:downlink_radar}.   
	\begin{figure} [t]
	\centering
	\includegraphics[width= \linewidth]{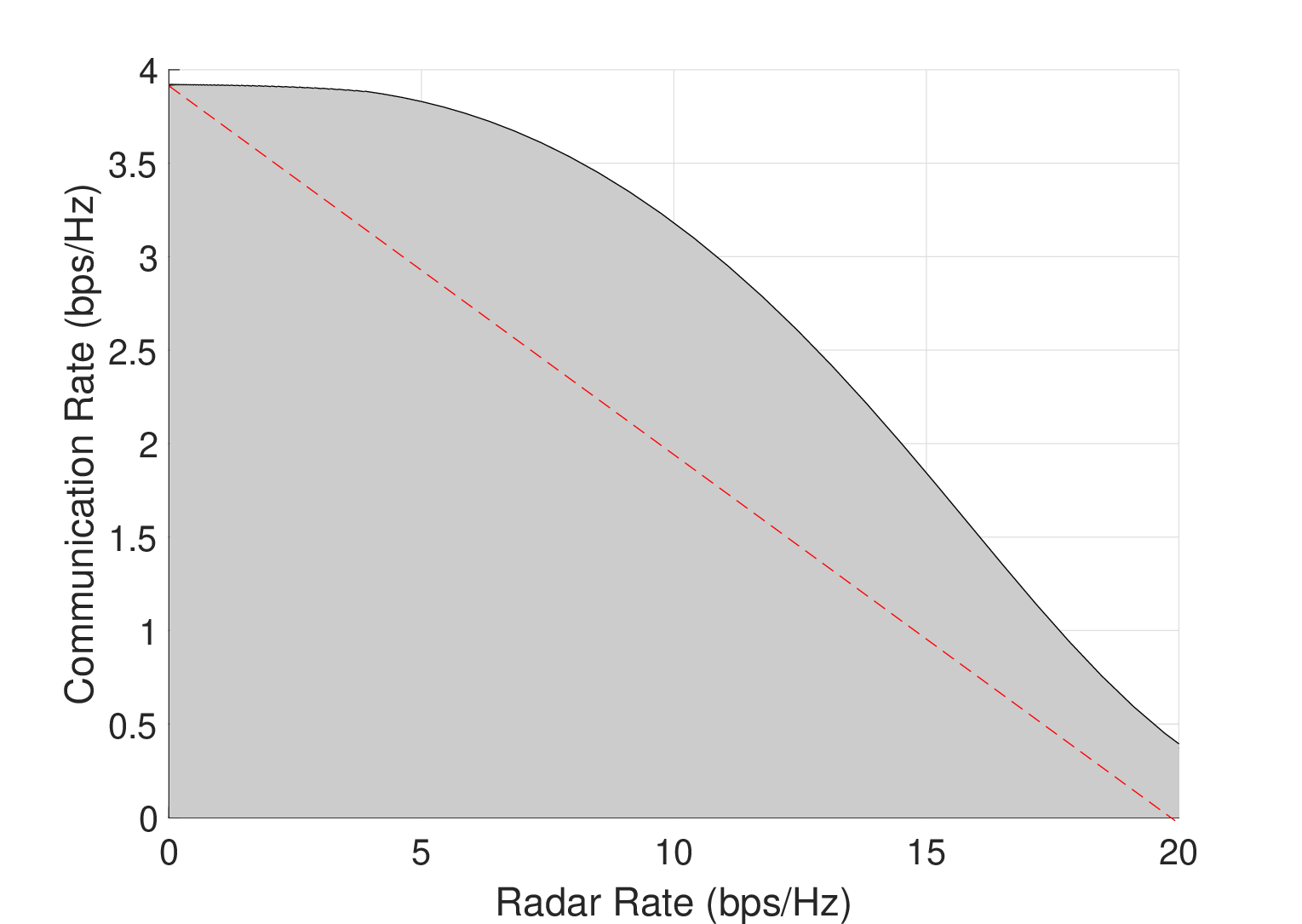}
	\caption{Consolidated Rate Region during the downlink subframe.}
	\label{f10}
\end{figure}

In Fig.~\ref{f10} we plot the achievable rate region for JRC system during the downlink subframe, \textcolor{black}{with the dotted line representing the case where the two systems are operated in the time division duplexed mode. Similar to the uplink case, this rate region is also convex, validating our hypothesis about the ability of massive MIMO systems to coexist with radars.}  
\section{Conclusions and Future Work}
In this paper, we investigated the performance of JRC system in which both  communication and radar sub-systems were operating simultaneously over same spectrum. To evaluate the performance of this system, we modelled it as a multiple access channel with both the subsystems non-cooperatively contending for the available resources. Following this, using results from random matrix theory, and via extensive simulations, we obtained the achievable rate regions for the system considering both uplink and downlink data transmission in the communication subsystem. We have observed these rate regions to be sufficiently convex and can safely conclude that massive MIMO systems can coexist with MIMO radars without any significant co-design. Future work may include the extension of this work to a cell free setting, \textcolor{black}{and the design of power control techniques for performance optimization}. 
\appendix
\subsection{Proof of Theorem~\ref{thm:uplink_comm}} \label{pr:thm_1}
Treating interference as noise~\cite{Hassibi_TIT_2003}, we can write the SINR of the received signal as the ratio of the mean squared value of the desired component to the sum of the mean squared values of all other components. In this context, we can write,   
\begin{equation}
\begin{split}
\zeta_{s,k}
= \beta_k\epsilon_{u,s,k} E\{|\hat{\mathbf{h}}^{H}_{k} \mathbf{R}_{yy|\hat{\mathbf{G}}_{rb},\hat{\mathbf{H}}}^{-1} \hat{\mathbf{h}}_{k}|^{2}  \}\\
\end{split}
\end{equation}
Now,
\begin{equation}
\hat{\mathbf{h}}^{H}_{k} \mathbf{R}_{yy|\hat{\mathbf{G}}_{rb},\hat{\mathbf{H}}}^{-1} \hat{\mathbf{h}}_{k}=\frac{\hat{\mathbf{h}}^{H}_{k} \mathbf{R}_{yy|\hat{\mathbf{G}}_{rb},\hat{\mathbf{H}}_{k}}^{-1} \hat{\mathbf{h}}_{k}}{1+\hat{\mathbf{h}}^{H}_{k} \mathbf{R}_{yy|\hat{\mathbf{G}}_{rb},\hat{\mathbf{H}}_{k}}^{-1} \hat{\mathbf{h}}_{k}}
\end{equation}
where $\hat{\mathbf{H}}_{k}  \in \mathcal{C}^{M \times (K-1)} $ contains all columns of $\mathbf{\hat{H}}$ except $\mathbf{\hat{h}}_{k}$. Thus, $\mathbf{R}_{yy|\hat{\mathbf{G}}_{rb},\hat{\mathbf{H}}_{k}}$ can be expressed as
{\color{black}\begin{multline}
\mathbf{R}_{yy|\hat{\mathbf{G}}_{rb},\hat{\mathbf{H}}_{k}}=\sum_{\substack{m=1, \\ m \neq k}}^{K} \beta_m \epsilon_{u,s,m}  \mathbf{\hat{h}}_m\mathbf{\hat{h}}_m^H+\sigma_r^2\hat{\mathbf{G}}_{rb}\hat{\mathbf{G}}_{rb}^H\\+\sum_{m=1}^{K} \beta_m \epsilon_{u,s,m} \bar{\mathbf{B}}^{2}_{m} +\left(\sigma_r^{2} \eta_{e} +N_{0}\right) \mathbf{I}_{M}.
\end{multline}}
Now, since $\hat{\mathbf{h}}_{k}$ and $\mathbf{R}_{yy|\hat{\mathbf{G}}_{rb},\hat{\mathbf{H}}_{k}}^{-1}$ are independent, 
{\color{black}\begin{equation}
\hat{\mathbf{h}}^{H}_{k} \mathbf{R}_{yy|\hat{\mathbf{G}}_{rb},\hat{\mathbf{H}}_{k}}^{-1} \hat{\mathbf{h}}_{k} \underset{M \rightarrow \infty}{\overset{a.s.}\longrightarrow} \text{Tr}\{\mathbf{R}_{yy|\hat{\mathbf{G}}_{rb},\hat{\mathbf{H}}_{k}}^{-1} \mathbf{B}_{k} \}.
\label{eq:mu_k}
\end{equation}}
Letting {\color{black}$\rho=\sigma_r^{2} \eta_{e} +N_{0}$, $\mathbf{S}= \sum_{m=1}^{K} \beta_m \epsilon_{u,s,m} \bar{\mathbf{B}}^{2}_{m} $} and $\mathbf{D}_{k}$ be a diagonal matrix, such that  its $m$th diagonal element $d_{k,m}$ is given by $d_{k,m}=\beta_{m} \epsilon_{u,s,m}$ where $m \in \{1,2,\ldots,k-1,k+1,\dots,K\}$. We can now write, 
{\color{black} \begin{equation}
\mathbf{R}_{yy|\hat{\mathbf{G}}_{rb},\hat{\mathbf{H}}_{k}}=\hat{\mathbf{H}}_{k} \mathbf{D}_{k,\beta \epsilon_{u,s}}   \hat{\mathbf{H}}_{k}^{H}+\sigma_r^2\hat{\mathbf{G}}_{rb}\hat{\mathbf{G}}_{rb}^H+\mathbf{S}+ \rho \mathbf{I}_{M}.
\end{equation}}
Since $\mathbf{S} \in \mathcal{C}^{M \times M}$ is a non negative definite matrix, $\hat{\mathbf{H}}_{k}\in \mathcal{C}^{M \times (K-1)}$ is a random matrix  and $\rho > 0$, ~\cite{Papazafeiropoulos2015Deterministic},
{\color{black}\begin{equation}
\text{Tr}\{\mathbf{R}_{yy|\hat{\mathbf{G}}_{rb},\hat{\mathbf{H}}_{k}}^{-1} \mathbf{B}_{k} \} \underset{M \rightarrow \infty}{\overset{a.s.}\longrightarrow} \mu_{k}.
\end{equation}}
Back substituting the expression for $\zeta_{s,k}$ results in~\eqref{eq:zeta_uxk}.
Now, 
\begin{equation}
\zeta_{I,k}= \left(  \sum_{\substack{l=1, \\ l \neq k}}^{K} \beta_l \epsilon_{u,s,l} E \{|\hat{\mathbf{h}}^{H}_{k} \mathbf{R}_{yy|\hat{\mathbf{G}}_{rb},\hat{\mathbf{H}}}^{-1} \hat{\mathbf{h}}_{l} |^{2}  \}\right)
\label{eq:zeta_IK}
\end{equation}
\begin{equation}
|\hat{\mathbf{h}}^{H}_{k} \mathbf{R}_{yy|\hat{\mathbf{G}}_{rb},\hat{\mathbf{H}}}^{-1} \hat{\mathbf{h}}_{l}|^{2}=  \frac{|\hat{\mathbf{h}}^{H}_{k} \mathbf{R}_{yy|\hat{\mathbf{G}}_{rb},\hat{\mathbf{H}}_{k}}^{-1} \hat{\mathbf{h}}_{l}|^{2}}{|1 + \hat{\mathbf{h}}^{H}_{k} \mathbf{R}_{yy|\hat{\mathbf{G}}_{rb},\hat{\mathbf{H}}_{k}}^{-1} \hat{\mathbf{h}}_{k}|^{2}}, 
\end{equation}
and,
\begin{multline}
|\hat{\mathbf{h}}^{H}_{k} \mathbf{R}_{yy|\hat{\mathbf{G}}_{rb},\hat{\mathbf{H}}_{k}}^{-1} \hat{\mathbf{h}}_{l}|^{2}   =  \left|\hat{\mathbf{h}}^{H}_{k} \mathbf{R}_{yy|\hat{\mathbf{G}}_{rb},\hat{\mathbf{H}}_{k,l}}^{-1} \hat{\mathbf{h}}_{l}- \right. \\ \left.\frac{ \hat{\mathbf{h}}^{H}_{k} \mathbf{R}_{yy|\hat{\mathbf{G}}_{rb},\hat{\mathbf{H}}_{k,l}}^{-1} \hat{\mathbf{h}}_{l} \hat{\mathbf{h}}^{H}_{l} \mathbf{R}_{yy|\hat{\mathbf{G}}_{rb},\hat{\mathbf{H}}_{k,l}}^{-1} \hat{\mathbf{h}}_{k}}{1+ \hat{\mathbf{h}}^{H}_{l}  \mathbf{R}_{yy|\hat{\mathbf{G}}_{rb},\hat{\mathbf{H}}_{k,l}}^{-1} \hat{\mathbf{h}}_{l}} \right|^{2}, 
\end{multline}
where $\hat{\mathbf{H}}_{k,l}\in \mathcal{C}^{M \times (K-2)}$ contains all columns of $\mathbf{\hat{H}}$ except $\mathbf{\hat{h}}_{k}$ and $\mathbf{\hat{h}}_{l}$. But,
{\color{black}\begin{equation}
\mathbf{R}_{yy|\hat{\mathbf{G}}_{rb},\hat{\mathbf{H}}_{k,l}}=\hat{\mathbf{H}}_{k,l} \mathbf{D}_{k,l,\beta \epsilon_{u,s}}   \hat{\mathbf{H}}_{k,l}^{H}+\sigma_r^2\hat{\mathbf{G}}_{rb}\hat{\mathbf{G}}_{rb}^H+\mathbf{S}+\rho \mathbf{I}_{M},
\end{equation}}
where $\mathbf{D}_{k,l,\beta \epsilon_{u,s}} $  is a diagonal matrix of order $(K-2)$ having the $m$th diagonal entry as $\beta_{m} \epsilon_{u,s,m}$, $m\neq\{k,l\}$. Since, $\hat{\mathbf{h}}^{H}_{k}, \mathbf{R}_{yy|\hat{\mathbf{G}}_{rb},\hat{\mathbf{H}}_{k,l}}^{-1}$ and $ \hat{\mathbf{h}}_{l}$ are independent, we have   
{\color{black}
\begin{equation}
|\hat{\mathbf{h}}^{H}_{k} \mathbf{R}_{yy|\hat{\mathbf{G}}_{rb}, \hat{\mathbf{H}}_{k,l}}^{-1} \hat{\mathbf{h}}_{l}|^{2} \underset{M \rightarrow \infty}{\overset{a.s.}\longrightarrow} \text{Tr}\{ \mathbf{R}_{yy|\hat{\mathbf{G}}_{rb},\hat{\mathbf{H}}_{k,l}}^{-1} \mathbf{B}_{k} \mathbf{R}_{yy|\hat{\mathbf{G}}_{rb},\hat{\mathbf{H}}_{k,l}}^{-1} \mathbf{B}_{l}  \}.
\end{equation}}
Now,
{\color{black}
\begin{equation}
\text{Tr}\{ \mathbf{R}_{yy|\hat{\mathbf{G}}_{rb},\hat{\mathbf{H}}_{k,l}}^{-1} \mathbf{B}_{k} \mathbf{R}_{yy|\hat{\mathbf{G}}_{rb},\hat{\mathbf{H}}_{k,l}}^{-1}  \mathbf{B}_{l} \} \underset{M \rightarrow \infty}{\overset{a.s.}\longrightarrow} \text{Tr}\{\mathbf{B}_{l} \mathbf{T}^{'}_{k,l}(\rho) \},
\end{equation}
\begin{equation}
\hat{\mathbf{h}}^{H}_{l}  \mathbf{R}_{yy|\hat{\mathbf{G}}_{rb}, \hat{\mathbf{H}}_{k,l}}^{-1} \hat{\mathbf{h}}_{l} \underset{M \rightarrow \infty}{\overset{a.s.}\longrightarrow}  \text{Tr}\{ \mathbf{R}_{yy|\hat{\mathbf{G}}_{rb}, \hat{\mathbf{H}}_{k,l}}^{-1} \mathbf{B}_{l}\},
\end{equation}
\begin{equation}
\text{Tr}\{ \mathbf{R}_{yy|\hat{\mathbf{G}}_{rb}, \hat{\mathbf{H}}_{k,l}}^{-1} \mathbf{B}_{l}\} \underset{M \rightarrow \infty}{\overset{a.s.}\longrightarrow} \text{Tr}\{\mathbf{T}_{k,l}(\rho) \mathbf{B}_{l} \}.
\end{equation}}
Hence,
\begin{equation}
\hat{\mathbf{h}}^{H}_{l}  \mathbf{R}_{yy|\hat{\mathbf{G}}_{rb}, \hat{\mathbf{H}}_{k,l}}^{-1} \hat{\mathbf{h}}_{l} \underset{M \rightarrow \infty}{\overset{a.s.}\longrightarrow}  \mu_{k,l},
\end{equation}
we can hence write,
\begin{multline}
|\hat{\mathbf{h}}^{H}_{k} \mathbf{R}_{yy|\hat{\mathbf{G}}_{rb},\hat{\mathbf{H}}}^{-1} \hat{\mathbf{h}}_{l}|^{2}=\frac{1}{|1+\mu_{k}|^{2}}\Bigg(  \mu^{'}_{k,l}+\frac{(\mu^{'}_{k,l})^2}{|1+
	\mu_{k,l}|^{2}} \\- 2 \Re\left\{\frac{( \mu^{'}_{k,l})^{3/2}}{1+ \mu_{k,l}}\right\}\Bigg),
\end{multline} 
back substituting the expressions in \eqref{eq:zeta_IK} to obtain \eqref{eq:zeta_IKR}.
\begin{equation}
\zeta_{E,k}= \sum_{l=1}^{K}   \beta_l\epsilon_{u,s,l} E\{|\hat{\mathbf{h}}^{H}_{k} \mathbf{R}_{yy|\hat{\mathbf{G}}_{rb},\hat{\mathbf{H}}}^{-1}\mathbf{\tilde{h}}_{l} x_l[n]|^{2}\},
\end{equation}
where,
\begin{equation}
|\hat{\mathbf{h}}^{H}_{k} \mathbf{R}_{yy|\hat{\mathbf{G}}_{rb},\hat{\mathbf{H}}}^{-1}\mathbf{\tilde{h}}_{l}|^{2}=\left|\frac{\hat{\mathbf{h}}^{H}_{k} \mathbf{R}_{yy|\hat{\mathbf{G}}_{rb},\hat{\mathbf{H}}_{k}}^{-1}\mathbf{\tilde{h}}_{l}}{1+\hat{\mathbf{h}}^{H}_{k} \mathbf{R}_{yy|\hat{\mathbf{G}}_{rb},\hat{\mathbf{H}}_{k}}^{-1}\mathbf{\hat{h}}_{k}}\right|^{2},
\end{equation}
\begin{equation}
|\hat{\mathbf{h}}^{H}_{k} \mathbf{R}_{yy|\hat{\mathbf{G}}_{rb}, \hat{\mathbf{H}}_{k}}^{-1} \tilde{\mathbf{h}}_{l}|^{2} \underset{M \rightarrow \infty}{\overset{a.s.}\longrightarrow}  \text{Tr}\{ \mathbf{R}_{yy|\hat{\mathbf{G}}_{rb},\hat{\mathbf{H}},k}^{-1} \mathbf{B}_{k} \mathbf{R}_{yy|\hat{\mathbf{G}}_{rb},\hat{\mathbf{H}},k}^{-1}  \bar{\mathbf{B}}^{2}_{l}  \},
\end{equation}
\begin{equation}
\text{Tr}\{ \mathbf{R}_{yy|\hat{\mathbf{G}}_{rb},\hat{\mathbf{H}},k}^{-1} \mathbf{B}_{k} \mathbf{R}_{yy|\hat{\mathbf{G}}_{rb},\hat{\mathbf{H}},k}^{-1}  \bar{\mathbf{B}}^{2}_{l} \} \underset{M \rightarrow \infty}{\overset{a.s.}\longrightarrow} \text{Tr}\{\mathbf{T}^{'}_{k}(\rho) \bar{\mathbf{B}}^{2}_{l}  \},
\label{e10}
\end{equation}
we can use these values to obtain~\eqref{eq:zeta_EkF}.
\begin{equation}
\zeta_{RC,k}=\sum_{i=1}^{N_t} \sigma^{2}_r |\hat{\mathbf{h}}^{H}_{k} \mathbf{R}_{yy|\hat{\mathbf{G}}_{rb},\hat{\mathbf{H}}}^{-1}\mathbf{\hat{g}}_{rb,i}|^{2},
\end{equation}
where
\begin{equation}
\hat{\mathbf{h}}^{H}_{k} \mathbf{R}_{yy|\hat{\mathbf{G}}_{rb},\hat{\mathbf{H}}}^{-1}\mathbf{\hat{g}}_{rb,i}=\frac{\hat{\mathbf{h}}^{H}_{k} \mathbf{R}_{yy|\hat{\mathbf{G}}_{rb},\hat{\mathbf{H}}_{k}}^{-1} \hat{\mathbf{g}}_{rb,i}}{1+\hat{\mathbf{h}}^{H}_{k} \mathbf{R}_{yy|\hat{\mathbf{G}}_{rb},\hat{\mathbf{H}}_{k}}^{-1} \hat{\mathbf{h}}_{k}},
\end{equation}
and,
\begin{multline}
|\hat{\mathbf{h}}^{H}_{k} \mathbf{R}_{yy|\hat{\mathbf{G}}_{rb},\hat{\mathbf{H}}_{k}}^{-1} \hat{\mathbf{g}}_{rb,i}|^{2}   =  \left|\hat{\mathbf{h}}^{H}_{k} \mathbf{R}_{yy|\hat{\mathbf{G}}_{rb,i},\hat{\mathbf{H}}_{k}}^{-1} \hat{\mathbf{g}}_{rb,i}- \right. \\ \left.\frac{ \hat{\mathbf{h}}^{H}_{k} \mathbf{R}_{yy|\hat{\mathbf{G}}_{rb,i},\hat{\mathbf{H}}_k}^{-1} \hat{\mathbf{g}}_{rb,i} \hat{\mathbf{g}}^{H}_{rb,i} \mathbf{R}_{yy|\hat{\mathbf{G}}_{rb,i},\hat{\mathbf{H}}_k}^{-1} \hat{\mathbf{h}}_{k}}{1+ \hat{\mathbf{g}}^{H}_{rb,i}  \mathbf{R}_{yy|\hat{\mathbf{G}}_{rb,i},\hat{\mathbf{H}}_k}^{-1} \hat{\mathbf{g}}_{rb,i}} \right|^{2}, 
\end{multline}
\begin{equation}
\mathbf{R}_{yy|\hat{\mathbf{G}}_{rb,i},\hat{\mathbf{H}}_{k}}=\hat{\mathbf{H}}_{k} \mathbf{D}_{k,\beta \epsilon_{u,s}}   \hat{\mathbf{H}}_{k}^{H}+\sigma_r^2\hat{\mathbf{G}}_{rb,i}\hat{\mathbf{G}}_{rb,i}^H+\mathbf{S}+\rho \mathbf{I}_{M},
\end{equation}
where $\hat{\mathbf{G}}_{rb,i}$ contains all the columns of $\hat{\mathbf{G}}_{rb}$ except $\hat{\mathbf{g}}_{i}$. Since, $\hat{\mathbf{h}}^{H}_{k}, \mathbf{R}_{yy|\hat{\mathbf{G}}_{rb,i}\hat{\mathbf{H}}_k}^{-1}$ and $ \hat{\mathbf{g}}_{rb,i}$ are independent, we have   
{\color{black}
\begin{multline}
|\hat{\mathbf{h}}^{H}_{k} \mathbf{R}_{yy|\hat{\mathbf{G}}_{rb,i} \hat{\mathbf{H}}_{k}}^{-1} \hat{\mathbf{g}}_{rb,i}|^{2} \underset{M \rightarrow \infty}{\overset{a.s.}\longrightarrow}   \text{Tr}\{ \mathbf{R}_{yy|\hat{\mathbf{G}}_{rb,i} \hat{\mathbf{H}}_{k}}^{-1} \mathbf{B}_{k} \times \\ \mathbf{R}_{yy|\hat{\mathbf{G}}_{rb,i} \hat{\mathbf{H}}_{k}}^{-1}  (\eta_{I}-\eta_{e}) \mathbf{I}_{M}  \} ,
\end{multline}
\begin{multline}
\text{Tr}\{ \mathbf{R}_{yy|\hat{\mathbf{G}}_{rb,i} \hat{\mathbf{H}}_{k}}^{-1} \mathbf{B}_{k} \mathbf{R}_{yy|\hat{\mathbf{G}}_{rb,i} \hat{\mathbf{H}}_{k}}^{-1} (\eta_{I}-\eta_{e}) \mathbf{I}_{M}   \} \underset{M \rightarrow \infty}{\overset{a.s.}\longrightarrow} \\ \text{Tr}\{\mathbf{T}^{'}_{k,i}(\rho) (\eta_{I}-\eta_{e}) \mathbf{I}_{M}  \},
\end{multline}}
\begin{equation}
\hat{\mathbf{g}}^{H}_{rb,i}  \mathbf{R}_{yy|\hat{\mathbf{G}}_{rb,i},\hat{\mathbf{H}}_k}^{-1} \hat{\mathbf{g}}_{rb,i} \underset{M \rightarrow \infty}{\overset{a.s.}\longrightarrow}  \text{Tr}\{ \mathbf{R}_{yy|\hat{\mathbf{G}}_{rb,i},\hat{\mathbf{H}}_k}^{-1} (\eta_{I}-\eta_{e}) \mathbf{I}_{M} \},
\end{equation}
\begin{equation}
\text{Tr}\{ \mathbf{R}_{yy|\hat{\mathbf{G}}_{rb,i},\hat{\mathbf{H}}_k}^{-1} (\eta_{I}-\eta_{e}) \mathbf{I}_{M} \} \underset{M \rightarrow \infty}{\overset{a.s.}\longrightarrow} \text{Tr}\{\mathbf{T}_{k,i}(\rho) (\eta_{I}-\eta_{e}) \mathbf{I}_{M}  \},
\end{equation}
\begin{multline}
|\hat{\mathbf{h}}^{H}_{k} \mathbf{R}_{yy|\hat{\mathbf{G}}_{rb},\hat{\mathbf{H}}_{k}}^{-1} \hat{\mathbf{g}}_{rb,i}|^{2}=\mu^{'}_{k,i} \\ +  \frac{(b^{2}_{k} \mu^{'}_{k,i})^{2}}{|1+\mu_{k,i}|^{2}}-2 \Re \left\{\frac{(b^{2}_{k} \mu^{'}_{k,i})^{3/2}}{1+\mu_{k,i}}\right\},
\label{eq:zeta_Ekc}
\end{multline}
we can use these values to obtain~\eqref{eq:zeta_RCkF}.
Finally, \begin{multline}
\zeta_{RE,k}= \sum_{i=1}^{N_{t}}E\left[|\hat{\mathbf{h}}^{H}_{k} \mathbf{R}_{yy|\hat{\mathbf{G}}_{rb},\hat{\mathbf{H}}}^{-1}\mathbf{\tilde{g}}_{rb,i}{s_{i}}[n]|^{2} \right],\\
\label{eq:zeta_rekP} 
\end{multline}
and 
\begin{equation}
\zeta_{w,k}=N_0 E\left\{|\hat{\mathbf{h}}^{H}_{k} \mathbf{R}_{yy|\hat{\mathbf{G}}_{rb},\hat{\mathbf{H}}}^{-1}\mathbf{w}_b|^{2}\right\},
\end{equation}
can be simplified using techniques similar to the ones used for $\zeta_{E,k}$.
\subsection{Proof of Theorem~2}\label{pr:thm_2}
Let $\mathbf{Z}=[\mathbf{z}[1],\mathbf{z}[2],\dots,\mathbf{z}[N] ]$ such that $\mathbf{Z}$ consists of iid columns with $\mathbf{z}[n]\sim \mathcal{CN}( h_{rr} \mathbf{A}(\theta) \mathbf{s}[n],\sigma^{2}_{wr} \mathbf{I}_{N_{r}})$ where $\sigma^{2}_{wr} = \sum_{k=1}^{K} \epsilon_{u,p,k} \eta_{rk} +N_{0}$. Now, the log-likelihood function of $\mathbf{Z}$ is as given
\begin{multline}
	\ln(f_{\mathbf{Z}}(\mathbf{z}))=-N_{r}N \ln(\pi \sigma^{2}_{wr}){-\frac{1}{\sigma^{2}_{wr}} \sum_{n=1}^{N}||\mathbf{z}_{n}||^{2}} \\ + {\frac{2}{\sigma^{2}_{wr}}\sum_{n=1}^{N}\Re \{h_{rr}^{*}  \mathbf{s}^{H}[n] \mathbf{A}^{H}(\theta) \mathbf{z}_{n} \}}  {-\frac{1}{\sigma^{2}_{wr}}\sum_{n=1}^{N}||h_{rr} \mathbf{A}(\theta) \mathbf{s}[n]||^{2}}
	\\=
	\ln(f_{\mathbf{Z}}(\mathbf{z}))=-N_{r}N \ln(\pi \sigma^{2}_{wr}){-\frac{1}{\sigma^{2}_{wr}} \sum_{n=1}^{N}||\mathbf{z}_{n}||^{2}} \\ + {\frac{2}{\sigma^{2}_{wr}}\Re \{h_{rr}^{*}  \sum_{n_{t}=1}^{N_{t}} \mathbf{A}_{:n_{t}}^{H}(\theta) 	\mathbf{\Omega}_{n_{t}} \}}  {-\frac{1}{\sigma^{2}_{wr}}\sum_{n=1}^{N}||h_{rr} \mathbf{A}(\theta) \mathbf{s}[n]||^{2}},
\end{multline}
where $\mathbf{A}_{:n_{t}}(\theta)$ denotes the $n_{t}$th column of $\mathbf{A}(\theta)$ and $	\mathbf{\Omega}_{n_{t}} $ is the $n_{t}$th sufficient statistic defined as $
\mathbf{\Omega}_{n_{t}}=\sum_{n=1}^{N}\mathbf{z}[n]s^{*}_{n_{t}}[n] , n_{t}=1,2,\ldots,N_{t}$;
that is obtained by multiplying the observed data with the $n_{t}$th transmitted signal. The sufficient statistic matrix $\mathbf{\Omega}$ can now be expressed as 
\begin{equation}
	\mathbf{\Omega}=\text{vec}[\mathbf{\Omega}_{1},\mathbf{\Omega}_{2},\dots,\mathbf{\Omega}_{N_{t}}]
	=h_{rr}\mathbf{d}(\theta)+\mathbf{v}_{u},
\end{equation}
where $\mathbf{d}(\theta)=\text{vec}(\mathbf{A}(\theta))$, and $\mathbf{v}_{u}=\text{vec}(\sum_{n=1}^{N}\sum_{k=1}^{K}\sqrt{\epsilon_{u,p,k}}\mathbf{g}_{kr} \psi_k[n]\mathbf{s}^{H}[n]+\sum_{n=1}^{N}\sqrt{N}_0 \mathbf{w}_r[n]\mathbf{s}^{H}[n])$, such that $\mathbf{v}_{u} \sim \mathcal{CN}(\mathbf{0},(N_{0}\sigma_r^{2}+\sum_{k=1}^{K} \epsilon_{u,p,k} \eta_rk\sigma_r^{2})\mathbf{I}_{N_r N_t})$.

Now, the Fisher information~\cite{kay1993fundamentals} for estimating $\theta$ is given as
\begin{equation}
	J_{\theta \theta}=\frac{2\sigma_r^2|h_{rr}|^2}{N_0+\sum_{k=1}^{K} \epsilon_{u,p,k} \eta_rk}\Re\{\text{Tr}\{\mathbf{\dot{A}}(\theta)\mathbf{\dot{A}}^H(\theta)\}\},
\end{equation}
and hence the CRB for $\theta$ takes the form $\text{CRB}(\theta) =\frac{1}{J_{\theta\theta}} $.
\subsection{Proof of Theorem~\ref{thm:downlink_comm}} \label{pr:thm_3}
	We can write
	\begin{equation}
	\zeta_{r,k}=\beta_k \epsilon_{d,s,k} E[|\hat{\mathbf{h}}_{k}^{T}  (\bar{\mathbf{H}} \bar{\mathbf{H}}^{H} + \alpha \mathbf{I}_{M})^{-1} \hat{\mathbf{h}}_{k}^{*} |^{2}]
	\label{zeta_{r,k}}
	\end{equation}
	Following steps similar to the proof of Theorem 1, we get \eqref{Zeta_{r,k,t}}. 
	Similarly, \begin{equation}
	\zeta_{r,I,k}= \sum_{\substack{m=1, \\ m \neq k}}^{K} \beta_k \epsilon_{d,s,m} E\{ |\hat{\mathbf{h}}_{k}^{T}  (\bar{\mathbf{H}} \bar{\mathbf{H}}^{H} +  \alpha \mathbf{I}_{M})^{-1} \hat{\mathbf{h}}_{m}^{*} |^{2}\}
	\label{zeta_{r,I,k}}.
	\end{equation}
can be shown to reduce to \eqref{zeta_{r,I,k,t}}, and
	\begin{equation}
	\zeta_{r,E,k}=\sum_{l=1 }^{K} \beta_k \epsilon_{d,s,l}  E\{|\tilde{\mathbf{h}}_{k}^{T} (\bar{\mathbf{H}} \bar{\mathbf{H}}^{H} + \alpha \mathbf{I}_{M})^{-1} \hat{\mathbf{h}}_{l}^{*} |^{2}\},
	\label{zeta_{r,E,k}}
	\end{equation}
can be shown to reduce to \eqref{zeta_{r,E,k,t}}.
	\subsection{Proof of Theorem~\ref{thm:downlink_radar}}\label{pr:thm_4}
		Using results from~\cite{Papa_TVT_2016}we can show that,
	{\color{black}	\begin{multline}
		E[(\mathbf{\tilde{G}}_{br}^{H}\mathbf{Q}\text{diag} (\sqrt{\boldsymbol{\epsilon}_{d,s}}) \bar{\mathbf{p}}[n])(\mathbf{\tilde{G}}_{br}^{H}\mathbf{Q}\text{diag} (\sqrt{\boldsymbol{\epsilon}_{d,s}}) \bar{\mathbf{p}}[n])^{H}]  
\\ \underset{M \rightarrow \infty}{\overset{a.s.}\longrightarrow} \mu^{'}_{\alpha}
		\label{mu^{'}}.
		\end{multline}}
	and
	{\color{black}	\begin{multline}
		E[(\mathbf{\hat{g}}_{br,m}^{H}\mathbf{Q}_{\mathbf{\hat{g}}_{br,m}}\text{diag} (\sqrt{\boldsymbol{\epsilon}_{d,s}})\\\times  \bar{\mathbf{p}}[n])(\mathbf{\hat{g}}_{br,m}^{H}\mathbf{Q}_{\mathbf{\hat{g}}_{br,m}}\text{diag} (\sqrt{\boldsymbol{\epsilon}_{d,s}}) \bar{\mathbf{p}}[n])^{H}]  	\underset{M \rightarrow \infty}{\overset{a.s.}\longrightarrow}
		\label{mu^{'}}.
\mu_{\mathbf{\hat{g}}_{br,m}}
\end{multline}}
		\begin{multline}
		E[|(\mathbf{\hat{g}}_{br,m}^{H}\mathbf{Q}\text{diag} (\sqrt{\boldsymbol{\epsilon}_{d,s}}) \bar{\mathbf{p}}[n])|^{2}]=\frac{\mu^{'}_{i,\alpha}}{|1+ \mu_{\mathbf{\hat{g}}_{br,m}}|^{H}}
		\end{multline}
		Letting $\mathbf{\tilde{z}}[n]=\mathbf{\hat{G}}_{br}\mathbf{Q}\text{diag} (\sqrt{\boldsymbol{\epsilon}_{d,s}}) \bar{\mathbf{p}}[n]+\mathbf{\tilde{G}}_{br}\mathbf{Q}\text{diag} (\sqrt{\boldsymbol{\epsilon}_{d,s}}) \bar{\mathbf{p}}[n] +\sqrt{N}_0 \mathbf{w}_r[n]$, represent the noise and interference, it is easy to show that,
		\begin{equation}
		E[\mathbf{z}[n]\mathbf{z}^{H}[n]] \triangleq \sigma_{wr}^2 \mathbf{I}_{N_{r}}  
		\end{equation}
		Let $\mathbf{Z}=[\mathbf{z}[1],\mathbf{z}[2],\dots,\mathbf{z}[N] ]$ Therefore, 
		\begin{multline}
		\ln(f_{\mathbf{Z}}(\mathbf{z}))=-N_{r}N \ln(\pi \sigma^{2}_{wr,d}){-\frac{1}{\sigma^{2}_{wr,d}} \sum_{n=1}^{N}||\mathbf{z}_{n}||^{2}} \\ + {\frac{2}{\sigma^{2}_{wr,d}}\Re \{h_{rr}^{*}  \sum_{n_{t}=1}^{N_{t}} \mathbf{A}_{:n_{t}}^{H}(\theta) 	\mathbf{\Omega}_{n_{t}} \}}  {-\frac{1}{\sigma^{2}_{wr,d}}\sum_{n=1}^{N}||h_{rr} \mathbf{A}(\theta) \mathbf{s}[n]||^{2}}.
		\end{multline}
		The sufficient statistic matrix $\mathbf{\Omega}_{d}$ can now be expressed as 
		\begin{equation}
		\mathbf{\Omega}_{d}=\text{vec}[\mathbf{\Omega}_{1,d},\mathbf{\Omega}_{2,d},\dots,\mathbf{\Omega}_{N_{t},d}]
		=h_{rr}\mathbf{d}(\theta)+\mathbf{v}_{d},
		\end{equation}
		where $\mathbf{d}(\theta)=\text{vec}(\mathbf{A}(\theta))$, and $\mathbf{v}_{d}=\text{vec}\left(\sum_{n=1}^{N}(\mathbf{\hat{G}}_{br}^{H}\mathbf{Q}\text{diag} (\sqrt{\boldsymbol{\epsilon}_{d,s}}) \bar{\mathbf{p}}[n]+\mathbf{\tilde{G}}_{br}^{H}\mathbf{Q}\text{diag} (\sqrt{\boldsymbol{\epsilon}_{d,s}}) \bar{\mathbf{p}}[n]) \times \right. \\ 
		\left. \mathbf{s}^{H}[n]  +\sum_{n=1}^{N}\sqrt{N}_0 \mathbf{w}_r[n]\mathbf{s}^{H}[n] \right)$, 
		such that $\mathbf{v}_{d} \sim \mathcal{CN}(\mathbf{0},\sigma_{wr,d}^2 \sigma_r^{2} \mathbf{I}_{N_r N_t})$.
		
		Therefore, the Fisher information~\cite{kay1993fundamentals} for estimating $\theta$ is given as
		\begin{equation}
		J_{\theta \theta}=\frac{2\sigma_r^2|h_{rr}|^2}{\sigma_{wr}^2}\Re\{\text{Tr}\{\mathbf{\dot{A}}(\theta)\mathbf{\dot{A}}^H(\theta)\}\},
		\end{equation}
		and hence the CRB for $\theta$ takes the form $\text{CRB}(\theta) =\frac{1}{J_{\theta\theta}} $.
 \bibliographystyle{ieeetran}
\bibliography{bibJournalList,Jour_RCC_mMIMO_v4}
\end{document}